\documentclass{sig-alternate-no-copyright}

\usepackage{graphicx}
\usepackage{balance}
\usepackage{url}
\usepackage{times}

\usepackage{color,soul}
\usepackage[linesnumbered,ruled,vlined]{algorithm2e}
\usepackage{dsfont} %\mathds{}
\usepackage{amssymb} %\mathbb{}
\usepackage[table]{xcolor}
\usepackage{subcaption}

\usepackage{todonotes}
\usepackage{amssymb}% http://ctan.org/pkg/amssymb
\usepackage{pifont}% http://ctan.org/pkg/pifont
\usepackage{array}%to define L for table word wrap
\newcolumntype{A}{>{\centering\arraybackslash}m{2.8cm}}
\newcolumntype{B}{>{\centering\arraybackslash}m{6cm}}
\newcolumntype{L}{>{\centering\arraybackslash}m{2cm}}
\newcommand{\eat}[1]{}

\newcommand{\db}{PlatoDB}
\newcommand{\tree}{segment tree} % Change next three lines together
\newcommand{\treeCap}{Segment Tree} % See above
\newcommand{\treeCapAll}{SEGMENT TREE} % See above

\newtheorem{theorem}{Theorem}

\newtheorem{example}{Example}

\newenvironment{compact_item}
{\setlength{\leftmargini}{2em}
\begin{itemize}
  \setlength{\labelsep}{1em}
  \setlength{\itemsep}{.1em}
  \setlength{\parskip}{0pt}
  \setlength{\parsep}{0pt}}
{\end{itemize}}
\begin{document}
% Copyright
%\setcopyright{acmcopyright}

% DOI
%\doi{10.475/123_4}

% ISBN
%\isbn{123-4567-24-567/08/06}

%Conference
%\conferenceinfo{PLDI '13}{June 16--19, 2013, Seattle, WA, USA}

%\acmPrice{\$15.00}

% --- Author Metadata here ---
%\conferenceinfo{WOODSTOCK}{'97 El Paso, Texas USA}
%\CopyrightYear{2007} % Allows default copyright year (20XX) to be over-ridden - IF NEED BE.
%\crdata{0-12345-67-8/90/01}  % Allows default copyright data (0-89791-88-6/97/05) to be over-ridden - IF NEED BE.
% --- End of Author Metadata ---

%\title{\db: Processing of Multi-argument Aggregate Queries with Error Guarantees }
%\title{\db: Processing Aggregate Queries with Multiple Arguments and Error Guarantees }
%\title{\db: Efficient Error-bounded Approximate Aggregation over Multiple Sensor Datasets}
\title{Efficient Approximate Query Answering over Sensor Data with Deterministic Error Guarantees\thanks{Supported by NSF BIGDATA 1447943.}}
%\title{Efficient Approximate OLAP for Big Sensor Data with Absolute Guaranteed Error Bounds}

\numberofauthors{6}
\author{
% 1st. author
\alignauthor Jaqueline Brito\\
       \affaddr{UC San Diego}\\
       \email{jabrito@cs.ucsd.edu}
\alignauthor Korhan Demirkaya\\
       \affaddr{UC San Diego}\\
       \email{kdemirka@cs.ucsd.edu}
\alignauthor Boursier Etienne\\
       \affaddr{UC San Diego}\\
       \email{eboursier@cs.ucsd.edu}
\and
\alignauthor Yannis Katsis\\
       \affaddr{UC San Diego}\\
       \email{ikatsis@cs.ucsd.edu}
\alignauthor Chunbin Lin\\
       \affaddr{UC San Diego}\\
       \email{chunbinlin@cs.ucsd.edu}
\alignauthor Yannis Papakonstantinou\\
       \affaddr{UC San Diego}\\
       \email{yannis@cs.ucsd.edu}
}

\maketitle

%\input{0.Todo}
%\newpage

\begin{abstract}
With the recent proliferation of sensor data, there is an increasing need for the efficient evaluation of analytical queries over multiple sensor datasets. The magnitude of such datasets makes exact query answering infeasible, leading researchers into the development of approximate query answering approaches. However, existing approximate query answering algorithms are not suited for the efficient processing of queries over sensor data, as they exhibit at least one of the following shortcomings: (a) They do not provide deterministic error guarantees, resorting to weaker probabilistic error guarantees that are in many cases not acceptable, (b) they allow queries only over a single dataset, thus not supporting the multitude of queries over multiple datasets that appear in practice, such as correlation or cross-correlation and (c) they support relational data in general and thus miss speedup opportunities created by the special nature of sensor data, which are not random but follow a typically smooth underlying phenomenon.  

To address these problems, we propose \db; a system that exploits the nature of sensor data to compress them and provide efficient processing of queries over multiple sensor datasets, while providing deterministic error guarantees. \db\ achieves the above through a novel architecture that (a) at data import time pre-processes each dataset, creating for it an intermediate hierarchical data structure that provides a hierarchy of summarizations of the dataset together with appropriate error measures and (b) at query processing time leverages the pre-computed data structures to compute an approximate answer and deterministic error guarantees for ad hoc queries even when these combine multiple datasets.

As a result of its novel architecture, \db\ exhibits speedups of 1-3 orders of magnitude compared to systems that use the entire sensor datasets to compute exact query answers during experiments performed on real sensor datasets. 

%In the era of  Internet of Things (IoT), sensor devices are equipped everywhere to serve many different monitoring applications, which produces big amount of sensor data. In order to explore the relationship of different sensor data instantly, issuing efficient OLAP (Online Analytical Processing) queries over the big sensor data becomes a critical requirement of the IoT. Executing such OLAP queries over original sensor data is inefficient due to the large volume of the data and the complex of the queries. To improve the performance, a widely applied technique is to relax the accuracy of answers by running queries on  (lossy) compressed sensor data, which has far fewer data points. However, there is no existing system can provide absolute error guarantees for OLAP queries. In this paper, we propose an approximate OLAP processing system called \db, which can provide  absolute guaranteed error bounds for approximate answers. More precisely, \db{} maintains a hierarchical tree structure for each sensor data, where each node maps to a sensor data segment and contains (i) a function referring the compression schema, and (ii) error measures. By accessing the hierarchical tree structures in the query processing time, \db{} provides approximate answers with different levels of error guarantees. Experimental results on two real-life datasets demonstrate the effectiveness and efficiency of \db.
\end{abstract}

\section{Introduction}
\label{sec:introduction}

The increasing affordability of sensors and storage has recently led to the proliferation of sensor data in a variety of domains, including transportation, environmental protection, healthcare, fitness, etc. These data are typically of high granularity and as a result have substantial storage requirements, ranging from a few GB to many TB. For instance, a Formula 1 produces 20GB of data during two 90-minute practice sessions~\footnote{http://www.zdnet.com/article/formula-1-racing-sensors-data-speed-and-the-internet-of-things/}, while a commercial aircraft may generate 2.5TB of data per day~\footnote{http://www.datasciencecentral.com/profiles/blogs/that-s-data-science-airbus-puts-10-000-sensors-in-every-single}.

The magnitude of sensor datasets creates a significant challenge when it comes to query evaluation. Running analytical queries over the data (such as finding correlations between signals), which typically involve aggregates, can be very expensive, as the queries have to access significant amounts of data. This problem becomes worse when queries combine in ad hoc ways multiple sensor datasets. For instance, consider a data analytics scenario, where a user wants to combine (a) a location dataset providing the location of users for different points in time (as recorded by their smartphone's GPS) and (b) an air pollution dataset recording the air quality at different points in time and space (as recorded by air quality sensors) to compute the average quality of air inhaled by each user over a certain time period\footnote{This is a real example encountered during the DELPHI project conducted at UC San Diego, which studied how health-related data about individuals, including large amounts of sensor data, can be leveraged to discover the determinants of health conditions \cite{katsis2013delphi}.}. Answering this query requires accessing all location and air pollution measurements in the time period of interest, which can be substantial for long periods. To solve this problem, researchers have proposed approximate query processing algorithms \cite{JermaineAPD07,AgarwalMPMMS13,WuOT10,BabcockDM04,PansareBJC11,PansareBJC11,PottiP15,LazaridisM01} that approximate the query result by looking at a subset of the data.

However, existing approaches have the following shortcomings when it comes to the query processing of multiple sensor data sets:
\begin{compact_item}
  \item \emph{Lack of deterministic error guarantees.} Most query approximation algorithms provide probabilistic error guarantees. While this is sufficient for some use cases, it does not cover scenarios where the user needs deterministic guarantees ensuring that the returned answer is within the specified error bounds. 
  \item \emph{Lack of support of queries over multiple datasets.} Many techniques, such as wavelets, provide error guarantees only for queries over a single dataset. The errors can be arbitrarily large for queries ranging over multiple datasets, as they are unaware of how multiple datasets interact with each other.
  \item \emph{Data agnosticism.} The majority of existing techniques works for relational data in general and does not leverage compression opportunities that come from the fact that sensor data are not random in nature but follow typically smooth continuous phenomena.
\end{compact_item}

To overcome the limitations, we design the \emph{\db} system, which leverages the nature of sensor data to compress them and provide efficient processing of analytical queries over multiple sensor datasets, while providing deterministic error guarantees. In a nutshell, \db\ operates as follows: When initiated, it preprocesses each time series dataset and builds for it a binary tree structure, which provides a hierarchy of summarizations of segments of the original time series. A node in the tree structure summarizes a segment of time series through two components: (i) a compression function estimating the data points in the segment, and (ii) error measures indicating the distance between the compressed segment and the original one. The lower level nodes refers to finer-grained segments and smaller errors. During runtime, \db\ takes as input an aggregate query over potentially multiple sensor datasets together with an error or time budget and utilizes the tree structure for each of the datasets involved in the query to obtain an approximate answer together with a deterministic error guarantee that satisfies the time/error budget.\\ 

\begin{figure*}[t]
\centering
\includegraphics[width=1\textwidth]{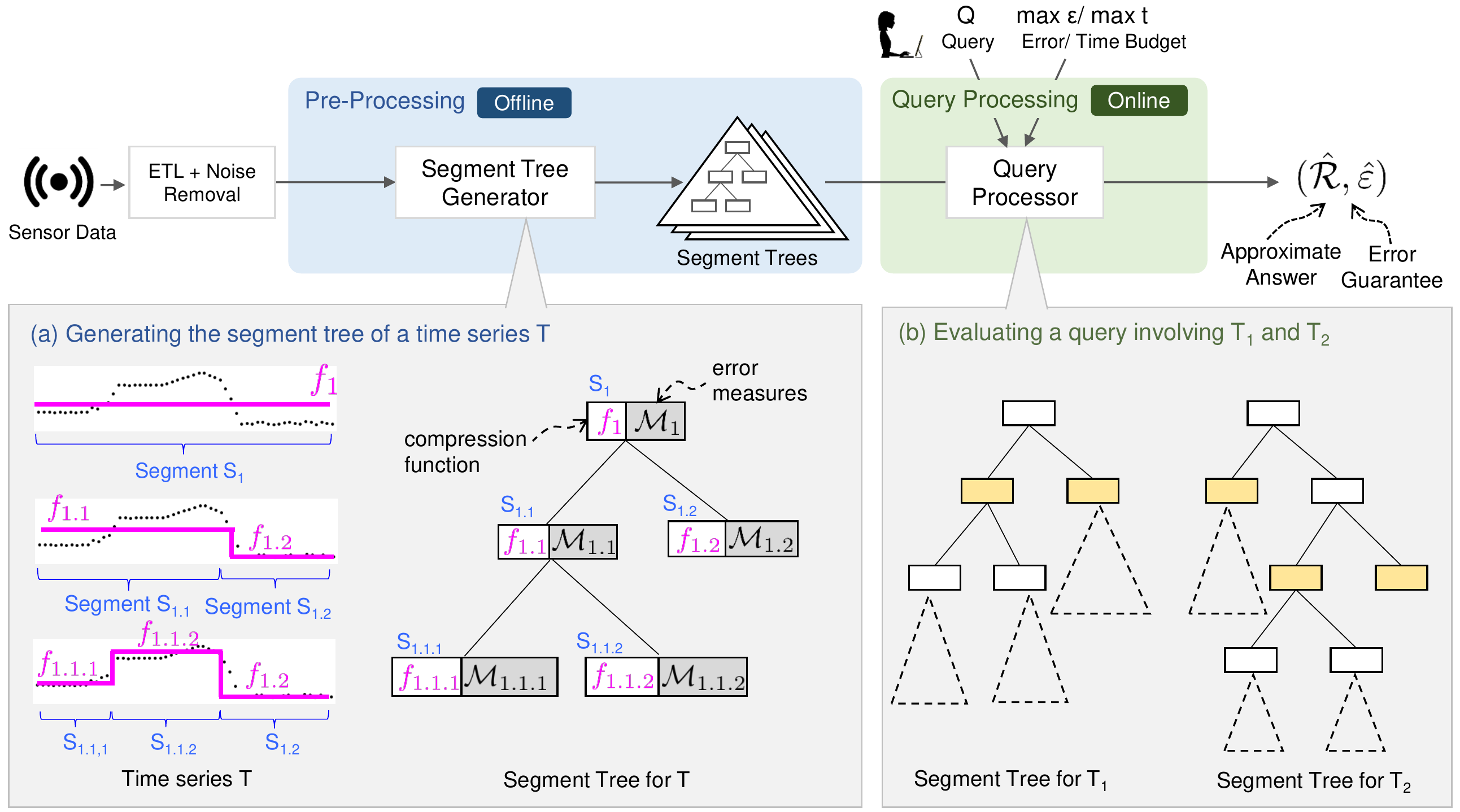}
\caption{\db's architecture, including details on the \tree\ generation and query processing.}
\label{fig:architecture}
\end{figure*}

\textbf{Contributions.} In this work, we make the following contributions:

\begin{itemize}
\item We define a query language over sensor data, which is powerful enough to express most common statistics over both single and multiple time series, such as variance, correlation, and cross-correlation (Section~\ref{sec:DataQueries}).
\item We propose a novel tree structure (structurally similar to hierarchical histograms) and a corresponding tree generation algorithm that provides a hierarchical summarization of each time series independently of the other time series. The summarization is based on the combination of arbitrary compression functions that can be reused from the literature together with three novel error measures that can be used to provide deterministic error guarantees, regardless of the employed compression function (Section~\ref{sec:index}). 
\item We design an efficient query processing algorithm operating on the pre-computed tree structures, which can provide deterministic error guarantees for queries ranging over multiple time series, even though each tree refers to one time series in isolation. The algorithm is based on a combination of error estimation formulas that leverage the error measures of individual time series segments to compute an error for an entire query (Section~\ref{Sec:error_estimator}) together with a tree navigation algorithm that efficiently traverses the time series tree to quickly compute an approximate answer that satisfies the error guarantees (Section~\ref{sec:QueryProcessing}).
\item We conduct experiments on two real-life datasets to evaluate our algorithms. The results show that our algorithm outperforms the baseline by 1-3 orders of magnitude (Section~\ref{experiments}).
\end{itemize}

\section{System Architecture}
\label{Sec:architecture}

Figure~\ref{fig:architecture} depicts \db's architecture. \db\ operates in two steps, performed at two different points in time. At data import time, \db\ pre-processes the incoming time series data, creating a \tree\ structure for each time series. At query execution time, it leverages these \tree s to provide an approximate query answer together with deterministic error guarantees. We next describe these two steps in detail.\\

\noindent \textbf{Off-line Pre-Processing.} At data import time, \db\ takes as input a set of time series. The time series are created from the raw sensor data by the typical Extract-Transform-Load (ETL) scripts potentially combined with de-noising algorithms, which is outside the focus of this paper.

For each such time series, \db's \emph{\treeCap\ Generator} creates a hierarchy of summarizations of the data in the form of a \emph{\tree}; a tree, whose nodes summarize the data for segments of the original time series. Intuitively, the structure of the \tree\ corresponds to a way of splitting the time series recursively into smaller segments: The root $S_{1}$ of the tree corresponds to the entire time series, which can be split into two subsegments (generally of different length), represented by the root's children $S_{1.1}$ and $S_{1.2}$. The segment corresponding to $S_{1.1}$ can be in turn split further into two smaller segments, represented by the children $S_{1.1.1}$ and $S_{1.1.2}$ of $S_{1.1}$ and so on. Since each node provides a brief summarization of the corresponding segment, lower levels of the tree provide a more precise representation of the time series than upper levels. As we will see later, this hierarchical structure of  segments is crucial for the query processor's ability to adapt to a wide variety of error/time budgets provided by the user. When the user is willing to accept a large error, the query processor will mostly use the top levels of the trees, providing a quick response. On the other hand, if the user demands a lower error, the algorithm will be able to satisfy the request by visiting lower levels of the \tree s (which exact nodes will be visited also depends on the query and the interplay of the time series in it). Leveraging the trees, \db\ can even provide users with continuously improving approximate answers and error guarantees, allowing them to stop the computation at any time, similar to works in online aggregation ~\cite{hellerstein1997online,CondieCAHES10,PansareBJC11}.

Each node of the tree summarizes the corresponding segment through two data items: (a) a \emph{compression function}, which represents the data points in a segment in a compact way (e.g., through a constant~\cite{KeoghCPM01} or a line~\cite{keogh1997fast}), and (b) a set of \emph{error measures}, which are metrics of the distance between the data point values estimated by the compression function and the actual values of the data points. As we will see, the query processor uses the compression function and error measures of the \tree\ nodes to produce an approximate answer of the query and the error guarantees, respectively. Interestingly, \db's internals are agnostic of the compression function used. As we will discuss in Section~\ref{sec:index}, \db's query processor works independently of the employed compression functions, allowing the system to be combined with all popular compression techniques. For instance, in our example above we utilized the Piecewise Aggregate Approximation (PAA) ~\cite{KeoghCPM01}, which returns the average of a set of values. However, we could have used other compression techniques, such as the  Adaptive Piecewise Constant Approximation (APCA)~\cite{keogh2001locally}, the Piecewise Linear Representation (PLR)~\cite{keogh1997fast}, or others.\\

\noindent\emph{Remark}. It is important to note that the \tree\ is not necessarily a balanced tree. \db\ decides whether a segment need to be split based on how close the values derived from the compression function are to the actual values of the segment. \db{} splits the segment when the difference is large. Intuitively, this means that the \tree\ contains more nodes for parts of the domain where the time series is irregular and/or rapidly changing, and fewer nodes for the smooth parts. \db{} treats the problem of finding the splitting positions as an \textit{optimization} problem, splitting at positions that can bring the largest error reduction. We will present the  \tree\ generator algorithms in Section \ref{sec:index}.

\begin{example}
Figure \ref{fig:architecture}(a) shows the \tree\ for a time series $T$. The root node $S_1$ of the tree (corresponding to the segment covering the entire time series) summarizes this segment through two items: a set of parameters describing a compression function $f_1$ (in this case the function returns the average $v$ of the values of the time series and can therefore be described by the single value $v$) and a set of error measures $M_1$ (the details of error measures will be presented in Section~\ref{sec:index}). This entire segment is split into two subsegments $S_{1.1}$ and $S_{1.2}$, giving rise to the identically-named tree nodes. Note that the tree is not balanced. Segment $S_{1.2}$ is not split further as its function $f_{1.2}$ correctly predicts the values within the corresponding segment. In contrast, the segment $S_{1.1}$ displays great variability in the time series' values and is thus split further into segments $S_{1.1.1}$ and $S_{1.1.2}$.
\end{example}

\noindent\textbf{On-line Query Processing.}
At query evaluation time, \db's \emph{Query Processor} receives a query and a time or error budget and leverages the pre-processed \tree s to produce an approximate query answer and a corresponding error guarantee satisfying the provided budget.

To compute the answer and error guarantee, \db{} traverses in parallel in a top-down fashion the \tree s of all time series involved in the query.  At any step of this process, it uses the compression function and error measures in the current accessed nodes to calculate an approximate query answer and the corresponding error. If it has not reached yet the time/error budget (i.e., if there is still time left or if the current error is still greater than the error budget), \db{} greedily chooses among all the currently accessed nodes the one, whose children nodes would yield the greatest error reduction and uses them to replace their parent in the answer and error estimation. Otherwise, \db{} stops accessing further nodes of the \tree s and outputs the currently computed approximate answer and error. Query processing is described in detail in Sections \ref{Sec:error_estimator} and \ref{sec:QueryProcessing}.\\
%\yannisk{We have not mentioned the incremental answer and error computation.}\\

\noindent\emph{Remark}. It is important to note that, in contrast to existing approximate query answering systems, \db\ can answer queries that span across different time series, even though the \tree s were pre-processed for each time series individually. As we will see, the fact that the \tree s were generated for each time series individually, leads to interesting problems at query processing time, such as aligning the segments of different time series and reasoning about how these segments interact to produce the query answer and error guarantees. Finally, it is also important to note that \db\ adapts to the provided error budget by accessing different number of nodes. Larger error budgets lead to fewer node accesses, while smaller error budgets require more node accesses.

\begin{example}
Consider a query $Q$ involving two time series $T_1$ and $T_2$ and an error budget $\varepsilon_{max}=10$. Figure~\ref{fig:architecture}(b) shows how the query processing algorithm uses the pre-computed \tree s of the two time series. \db{} first accesses the root nodes of both \tree s in parallel and computes the current approximate query answer $\hat{\mathds{R}}$ and error $\hat{\varepsilon}$, using the compression function and error measures in the root nodes. Let's assume that $\hat{\varepsilon}=20$. Since $\hat{\varepsilon}>\varepsilon_{max}$, \db{} keeps traversing the trees by greedily choosing a node and replacing it by its children, so that the error reduction at each step is maximized. This process continues until the error budget is satisfied. For instance, assume that using the yellow shaded nodes in Figure~\ref{fig:architecture}(b) \db{} obtains an error $\hat{\varepsilon}=6<\varepsilon_{max}$. Then \db{} stops traversing the trees and outputs the approximate answer and the error $\hat{\varepsilon}=6$. Note that none of the descendants of the shaded nodes is touched, resulting in big performance savings.
\end{example}

As a result of this architecture, \db\ achieves speedups of 1-3 orders of magnitude in query processing of sensor data compared to approaches that use the entire dataset to compute exact query answers (more details are included in \db's experimental evaluation in Section \ref{experiments}).

\section{Data and Queries}
\label{sec:DataQueries}
Before describing the \db\ system, we first present its data model and query language.\\
\begin{table*}
\centering
\begin{tabular}{|l|l|l|l|}
\hline
Statistic & Symbol & Definition & Query Expression\\
\hline
 Mean & $E(T)$ & $\sum\limits_{i=1}^nd_i$ & $\frac{Sum(T,1,n)}{n}$\\
\hline
Variance & $Var(T)$ & $\sum\limits_{i=1}^n(d_i-E(T))^2$ & $Sum(Times(T,T),1,n)-\frac{Sum(T,1,n)\times Sum(T,1,n)}{n}$\\
\hline
Covariance & $Cov(T_1,T_2)$ & $\frac{\sum\limits_{i=1}^n((d_i^{(1)}-E(T_1))(d_i^{(2)}-E(T_2)))}{n-1}$ & $\frac{Sum(Times(T_1,T_2),1,n)}{n-1}-\frac{Sum(T_1,1,n)\times Sum(T_2,1,n)}{n(n-1)}$\\
\hline
Correlation & $Corr(T_1,T_2)$ & $\frac{\sum\limits_{i=1}^n((d_i^{(1)}-E(T_1))(d_i^{(2)}-E(T_2))}{\sqrt{\sum\limits_{i=1}^n(d_i^{(1)}-E(T_1))^2\sum\limits_{i=1}^n(d_i^{(2)}-E(T_2))^2}}$ & $\frac{Sum(Times(T_1,T_2))-\frac{1}{n}Sum(T_1,1,n)\times Sum(T_2,1,n)}{\sqrt{Var(T_1)Var(T_2)}}$\\
\hline
Cross-correlation & $Coss(T_1,T_2,\ell)$ & $\frac{\sum\limits_{i=1}^n((d_i^{(1)}-E(T_1))(d_{i+\ell}^{(2)}-E(T_2))}{\sqrt{\sum\limits_{i=1}^n(d_i^{(1)}-E(T_1))^2\sum\limits_{i=1}^n(d_{i+\ell}^{(2)}-E(T_2))^2}}$ & $\frac{Sum(Times(T_1,T_2))-\frac{1}{n}Sum(T_1,1,n)\times Sum(T_2,1+l,n+l)}{\sqrt{Var(T_1)Var(T_2)}}$\\
\hline
\end{tabular}
\caption{Query expressions for common statistics.}
\label{tbl:common-statistics-queries}
\end{table*}

\noindent\textbf{Data Model.}
For the purpose of this work, a time series $T$=$[(t_1, d_1)$, $(t_2, d_2)$, $\ldots$, $(t_n, d_n)]$ is a sequence of (time, data point) pairs ($t_i$, $d_i$), such that the data point $d_i$ was observed at time $t_i$.  We follow existing work~\cite{GoldinK95} to normalize and standardize the time series so that all time series are in the same domain and have the same resolution. Since all time series are aligned, for ease of exposition we omit the exact time points and use instead the index of the data points whenever we need to define a time interval. For instance, we will denote the above time series simply as $T$=$(d_1, d_2,...,d_n)$, and use $[i, j]$ to refer to the time interval $[t_i, t_j]$. A subsequence of a time series is called a time series \textit{segment}. For example $S=(5.01, 5.06)$ is a segment of the time series $T=(5.05$, $\underline{5.01,5.06}$, $5.06$, $5.08)$.\\

\begin{figure}
\[
\begin{array}{|lll|c|}
\hline
\multicolumn{4}{|l|}{\cellcolor{gray!15} \text{Query Expression $(Q)$}}\\
\hline
\mathit{Q} & \rightarrow & Ar                      &\\
\hline
\multicolumn{4}{l}{}\\
\hline
\multicolumn{4}{|l|}{\cellcolor{gray!15} \text{Arithmetic Expression $(Ar)$}}\\
\hline
Ar & \rightarrow & \text{number} & \\
			   & |           & Agg	&\\
			   & |           & Ar \otimes Ar & \text{where } \otimes\in\{+,-,\times,\div\}\\
\hline
\multicolumn{4}{l}{}\\
\hline
\multicolumn{4}{|l|}{\cellcolor{gray!15} \text{Aggregation Expression $(Agg)$}}\\
\hline
Agg            & \rightarrow & \textsf{Sum}(T,\ell_s,\ell_e)   & \sum\limits_{i=\ell_s}^{\ell_e}d_{i}\\
\hline
\multicolumn{4}{l}{}\\
\hline
\multicolumn{4}{|l|}{\cellcolor{gray!15} \text{Time Series Expression $(T)$}}\\
\hline
T               & \rightarrow           & \text{base time series} & \\
                & |             & \textsf{SeriesGen}(\upsilon, n)    & (\underbrace{\upsilon,\upsilon,...,\upsilon}_{n})\\
               & |           & \textsf{Plus}(T, T)         & (d_1^{(1)} + d_1^{(2)}, \ldots, d_n^{(1)} + d_n^{(2)})\\
               & |           & \textsf{Minus}(T, T)        & (d_1^{(1)} - d_1^{(2)}, \ldots, d_n^{(1)} - d_n^{(2)})\\
               & |           & \textsf{Times}(T, T)        &(d_1^{(1)} * d_1^{(2)}, \ldots, d_n^{(1)} * d_n^{(2)})\\
\hline
\end{array}
\]
\caption{Grammar of query expressions.}
%$T_i$ is a time series $T_i = d_1^{(1)} - d_1^{(2)}, \ldots, d_n^{(1)} - d_n^{(2)}$$\upsilon$ is a constant and $\ell_s$, $\ell_e$, $n$ are integer constants. $\otimes\in\{+,-,\times,\div\}$.}
\label{fig:query}
\end{figure}

\noindent\textbf{Query Language.} \db\ supports queries whose main building blocks are aggregation queries over time series. Figure~\ref{fig:query} shows the formal definition of the query language and Table~\ref{tbl:common-statistics-queries} lists several common statistics that can be expressed in this language. 

A query expression $Q$ is an arithmetic expression of the form $Arr_1 \otimes Arr_2 \otimes \ldots Arr_n$, where $\otimes$ are the standard arithmetic operators ($+, - \times, \div$) and $Arr_i$ is either an arithmetic literal or an aggregation expression over a time series. An aggregation expression \textsf{Sum}$(T, l_s, l_e)$ over a time series $T$ computes the sum of all data points of $T$ in the time interval $[l_s, l_e]$. Note that the time series that is aggregated could either be a base time series or a derived time series that was computed from a set of base time series through a set of time series operators. \db\ allows a series of time series operators, including \textsf{Plus}$(T_1, T_2)$, $\textsf{Minus}(T_1, T_2)$, and $\textsf{Times}(T_1, T_2)$ (which return a time series that has data points computed by adding, subtracting, and multiplying the respective data points of the original time series, respectively), as well as $\textsf{SeriesGen}(v, n)$, which takes as input a value $v$ and a counter $n$ and creates a new time series that contains $n$ data points with the value $v$.

Note that the query language can be used to express many common statistics over time series encountered in practice and all the queries we encountered during the DELPHI project conducted at UC San Diego, which explored how health-related data about individuals, including large amounts of sensor data, can be leveraged to discover the determinants of health conditions and which served as the motivation for this work \cite{katsis2013delphi}. These include the mean and variance of a single time series, as well as the covariance, correlation, and cross-correlation between two time series. Table \ref{tbl:common-statistics-queries} shows how common statistics can be expressed in \db's query language.

\section{\treeCapAll}
\label{sec:index}

As explained in Section \ref{Sec:architecture}, at data import time, \db\ creates for each time series a hierarchy of summarizations of the series in the form of the \emph{\tree}. In this Section we first explain the structure of the tree and then describe the \tree\ generation algorithm.

\subsection{\treeCap\ Structure}
%\label{sec:error_estimation}

Let $T = (d_1, \ldots, d_n)$ be a time series. The \tree\ of $T$ is a binary tree whose
nodes summarize segments of the time series with nodes higher up the tree summarizing large segments and nodes lower down the tree summarizing progressively smaller segments. In particular, the root node summarizes the entire time series $T$. Moreover, for each node $n$ of the tree summarizing a segment $S_i = (d_i, \ldots, d_j)$ of $T$, its left and right children nodes $n_l$ and $n_r$ summarize two subsegments $S_l = (d_i, \ldots, d_k)$ and $S_r = (d_{k+1}, \ldots, d_j)$, respectively, which form a partitioning of the original segment $S_i$. As we will see in Section \ref{sec:QueryProcessing}, this hierarchical structure allows \db\ to adapt to varying error/time budgets by only accessing the parts of the tree required to achieve the given error/time budget.

At each node $n$ corresponding to segment $S_i = (d_i, \ldots, d_j)$, \db{} summarizes the segment $S_i$ by keeping two types of measures: (a) a description of a compression function that is used to approximately represent the time series values in the segment and (b) a set of error measures describing how far the above approximate values are from the real values. As we will see in Sections \ref{Sec:error_estimator} and \ref{sec:QueryProcessing}, \db\ uses at query processing time the compression function and error measures stored in each node to compute an approximate answer of the query and deterministic error guarantees, respectively. We next describe the compression functions and error measures stored within each \tree\ node in detail.\\

\noindent\textbf{Segment Compression Function.} Let $S = (d_1, \ldots, d_n)$ be a segment. \db\ summarizes its contents through a compression function $f$ used by the user.
\db\ supports the use of any of the compression functions suggested in the literature~\cite{KeoghCPM01,keogh2001locally,keogh1997fast,FaloutsosRM94,ChanF99,NgC04}. Examples include but are not limited to the Piecewise Aggregate Approximation (PAA) ~\cite{KeoghCPM01}, the Adaptive Piecewise Constant Approximation (APCA)~\cite{keogh2001locally}, the Piecewise Linear Representation (PLR)~\cite{keogh1997fast}, the Discrete Fourier Transformation (DFT)~\cite{FaloutsosRM94}, the Discrete Wavelet Transformation (DWT)~\cite{ChanF99}, and the Chebyshev polynomials (CHEB)~\cite{NgC04}.

To describe the function, \db\ stores in the segment node parameters describing the function. These parameters depend on the type of the function. For instance, if $f$ is a Piecewise Aggregate Approximation (PAA), estimating all values within a segment by a single value $b$, then the parameter is just a single value $b$. On the other hand, if $f$ is a Piecewise Linear Approximation (PLR), estimating the values in the segment through a line $ax+b$, then the function parameters are the coefficients $a$ and $b$ of the polynomial used to describe the line.

In the rest of the document, we will refer directly to the compression function $f$ (instead of the parameters that are used to describe it). Given a segment $(d_1, \ldots, d_n)$, we will use $f(i)$ to denote the value for element $d_i$ of the segment, as derived by $f$.\\

\noindent\textbf{Segment Error Measures.} In addition to the compression function, \db\ also stores a set of error measures for each time series segment $S = (d_1, \ldots, d_n)$. \db\ stores the following three error measures:

\begin{itemize}
  \item $\mathcal{L}$ : The sum of the absolute distances between the original and the compressed time series (also known as the Manhattan or $L_1$ distance), i.e., $\mathcal{L}=\sum\limits_{i=1}^{n}|d_i-f(i)|$.
  \item $d^*$ : The maximum absolute value of the original time series, i.e., $d^*=max\{|d_i|~|~ 1\leq i\leq n\}$.
  \item $f^*$ : The maximum absolute value of the compressed time series, i.e., $f^*=max\{|f(i)|~|~ 1\leq i\leq n\}$.
\end{itemize}

\begin{example}
For instance, consider a segment $S=(5.12,$ $5.09,5.07,5.04)$ summarized through the PAA compression function $f=5.08$ (i.e., $f(1)=f(2)=f(3)=f(4)=5.08$). Then $\mathcal{L}=|5.12-5.08|+|5.09-5.08|+|5.07-5.08|+|5.04-5.08|=0.1$, $d^*=max\{5.12,5.09,5.07,5.04\}=5.12$ and $f^*=max\{5.08,5.08,5.08,5.08\}=5.08$.
\end{example}

As we will see in Section \ref{Sec:error_estimator}, the above three error measures are sufficient to compute deterministic error guarantees for any query supported by the system, \emph{regardless} of the employed compression function $f$. This allows administrators to select the compression function best suited to each time series, without worrying about computing the error guarantees, which is automatically handled by \db.

\subsection{\treeCap\ Generation}

We next describe the algorithm generating the \tree. To build the tree, the algorithm has to decide how to build the children nodes from a parent node; i.e., how to partition a segment into two non-overlapping subsegments. Each possible splitting point will lead to different children segments and as a result to different errors when \db\ uses the children segments to answer a query at query processing time. Ideally, the splitting point should be the one that minimizes the error among all possible splitting points. However, since \db\ supports ad hoc queries and since each query may benefit from a different splitting point, there is no way for \db\ to choose a splitting point that is optimal for all queries.\\

\noindent\textbf{\treeCap\ Generation Algorithm.} Based on this observation, \db\ chooses the splitting point that minimizes the error for the basic query that simply computes the sum of all data points of the original segment. In particular, the \tree\ generation algorithm starts from the root and proceeding in a top-down fashion given a segment $S = (d_1, \ldots, d_n)$, selects a splitting point $d_k$ that leads into two subsegments $S_l = (d_1, \ldots, d_k)$ and $S_r = (d_{k+1}, \ldots, d_n)$ so that the sum of the Manhattan distances of the new subsegments $\mathcal{L}_{S_l} + \mathcal{L}_{S_r}$ is minimized.

The algorithm stops further splitting down a segment $S$, when one of the following two conditions hold: (i) When the Manhattan distance $\mathcal{L}_{S}$ of the segment is smaller than a threshold $\tau$ or (ii) when he size of the segment is below a threshold $\kappa$. The choice between conditions (i) and (ii) and the values of the corresponding thresholds $\tau$ and $\kappa$ is specified by the system administrator.

Since the algorithm needs time proportional to the size of a segment to compute the splitting point of a single segment and it repeats this process for every non-leaf tree node, it exhibits a worst-time complexity of $O(mn)$, where $n$ is the size of the original time series (i.e., the number of its data points) and  $m$ number of nodes in the resulting \tree.\\

%\yannisk{Unless I am mistaken, the complexity is not proportional to the height of the tree but to the total number of nodes. Chunbin, can you please double-check? Yes, you are correct. }\\

\noindent\textbf{Discussion.} Note that by deciding independently how to split each individual segment into two subsegments, the \tree\ generation algorithm is a greedy algorithm, which even though makes optimal local decisions for the basic aggregation query, may not lead to optimal global decisions. For instance, there is no guarantee that the $k$ nodes that exist at a particular level of the \tree\ correspond to the $k$ nodes that minimize the error of the basic aggregation query. The literature contains a multitude of algorithms that can provide such a guarantee for a given $k$; i.e., algorithms that can, given a time series $T$ and a number $k$, produce $k$ segments of $T$ that minimize some error metric. Examples include the optimal algorithm of \cite{bellman1961approximation}, as well as approximation algorithms with formal guarantees presented in \cite{terzi2006efficient}. However, all these algorithms have very high worst-time complexity that makes them prohibitive for the large number of data points typically found in sensor datasets and are therefore not considered in this work. Though several heuristic segmentation algorithms exist, such as the \textit{Sliding Windows}~\cite{ShatkayZ96}, the \textit{Top-down}~\cite{KeoghP98} and the \textit{Bottom-Up}~\cite{KeoghP99} algorithm, similar do our greedy algorithm, they do not provide any formal guarantees.

Finally, note that the tree generated by the above algorithm will in general be unbalanced. Intuitively, the algorithm will create more nodes and corresponding tree levels to cover segments that contain data points that are more irregular and/or rapidly changing, utilizing fewer nodes for smooth segments.

%\noindent\textbf{Discussion.} The top-down method is a greedy method, it finds the best pivot to partition segments in each step. It may not be optimal, but it is efficient. The optimality of the bottom-up method is determined by the employed segmentation algorithm. If the employed segmentation algorithm is optimal, then the tree created by the bottom-up method is optimal. However, creating optimal $k$ segments is inefficient, the complexity is $O(kn^2)$ where n is the number of data points in the entire time series. In short, the top-down method is efficient though it may not produce optimal structure, while the bottom-up method may produce optimal structure, it is inefficient.

\section{Computing Approximate Query Answers and Error Guarantees}
\label{Sec:error_estimator}

Given pre-computed segment trees for time series $T_1, \ldots, T_n$, \db\ answers ad hoc queries over the time series by accessing their segment trees. In particular, to answer a given query $Q$ under an error/time budget, \db\ navigates the segment trees of the time series involved in $Q$, selects segment nodes (or simply segments) that satisfy the budget, and computes an approximate answer for $Q$ together with deterministic error guarantees.

We will next present the query processing algorithm. For ease of exposition, we will start by describing how \db\ computes an approximate query answer and the associated error guarantees assuming that the segment nodes have been already chosen, and will explain in Section \ref{sec:QueryProcessing} how \db\ traverses the tree to choose the segment nodes.\\

\noindent\textbf{Approximate query answering problem under given segments.} Formally, let $T_1, \ldots, T_k$ be time series, such that time series $T_i$ is partitioned into segments $S_i^1, \ldots S_i^n$. Given (a) these segments and the associated measures as described above and (b) a query $Q$ over the time series $T_1, \ldots, T_k$, we will show how \db\ computes an approximate query answer $\hat{\mathds{R}}$ and an estimated error $\hat{\varepsilon}$, such that the approximate query answer $\hat{\mathds{R}}$ is guaranteed to be with $\pm\hat{\varepsilon}$ of the accurate query answer ${\mathds{R}}$\footnote{Accurate answer means running queries over raw data. But note that, in this work, we can given estimate errors wihout computing the accurate answers.}, i.e., $|\mathds{R}-\hat{\mathds{R}}|\leq\hat{\varepsilon}$.

For ease of exposition, we next first describe the simple case where each time series $T_i$ contains a single segment perfectly aligned with the single segment of the other series, before describing the general case, where each time series $T_i$ contains multiple segments, which may also not be perfectly aligned with the segments of the other time series.

%In the following, we first analyze the basic case that (1) each time series has only one segment, and (2) all the segments have equal sizes, i.e., they are perfectly aligned. Then we discuss the general case where each time series contains many segments and the sizes of segments are different.

\subsection{Single Time Series Segment}
\label{Sec:error_estimator_single}
Let $T_1, \ldots,T_k$ be $k$ time series with single aligned segments, i.e., $T_i$ is approximated by a single segment $S_i$. Also let $f_i$ be the compression function and $(\mathcal{L}_i, d_i^*, f_i^*)$ the error measures of segment $S_i$, respectively. To compute the approximate answer and error guarantees of a query $Q$ over $T_1, \ldots, T_k$ using the single segments $S_1, \ldots, S_k$, \db\ employs an algebraic approach computing in a bottom-up fashion for each algebraic operator $op$ of $Q$ the approximate answer and error guarantees for the subquery corresponding to the subtree rooted at $op$.

This algebraic approach is based on formulas that for each algebraic query operator, given an approximate query answer and error for the inputs of the operator, provide the corresponding query answer and error for the output of the operator. Figure \ref{fig:error_estimators} shows the formulas employed by \db\ for each algebraic query operator supported by the system. Note that the output signatures differ between operators. This is due to the different types of operators supported by \db, as explained next. Recall from Section \ref{sec:DataQueries} that \db's query language consists of three types of operators: (i) time series operators, (ii) aggregation operator, and (iii) arithmetic operators. While time series operators output a time series, aggregation and arithmetic operators output a single number. As a result, the formulas used for answer and error estimation, treat these two classes of operators differently: For time series operators, the formulas return, similarly to the input time series, the compression function and error measures of the output time series. For aggregation and arithmetic operators on the other hand, which return a single number and not an entire time series, the formulas return simply a single approximate answer and estimated error. Figure \ref{fig:error_estimators} shows the resulting formulas.~\footnote{Out of the formulas, the most involved are the output measure estimation formulas of the \textsf{Times} operator. More details on how they were derived can be found in Appendix \ref{appendix:times}.}

\begin{figure}[t]
\small
\raggedright
\textbf{Time Series Operators}
\vspace*{0.2cm}
\renewcommand{\tabcolsep}{0.7mm}
\begin{tabular}{|l||c||c|c|c|} \hline
\textbf{Operator}                 & \cellcolor{gray!25} \textbf{Compr.} &\multicolumn{3}{c|}{\cellcolor{gray!25} \textbf{Output}} \\
                                  & \cellcolor{gray!25} \textbf{Func.} &\multicolumn{3}{c|}{\cellcolor{gray!25} \textbf{Error Measures}} \\
\cline{2-5}
                                  & \cellcolor{gray!13} $f$ & \cellcolor{gray!13} $\mathcal{L}$  &   \cellcolor{gray!13}  $d^*$     &  \cellcolor{gray!13} $f^*$\\\hline
\textsf{SeriesGen($\upsilon, n$)} & $\upsilon$ &$0$            &    $\upsilon$& $\upsilon$ \\\hline
\textsf{Plus($T_1,T_2$)}          & $f_1+f_2$ & $\mathcal{L}_{1}+\mathcal{L}_{2}$ & $d^{*}_1+d^{*}_2$ & $f^{*}_1+f^{*}_2$\\\hline
\textsf{Minus($T_1,T_2$)}         & $f_1-f_2$ & $\mathcal{L}_{1}+\mathcal{L}_{2}$ & $d^{*}_1+d^{*}_2$ & $f^{*}_1+f^{*}_2$\\\hline
\textsf{Times($T_1,T_2$)}         & $f_1\times f_2$ & \multicolumn{1}{l|}{$min\{$} & $d^{*}_1\times d^{*}_2$ & $f^{*}_1\times f^{*}_2$\\
        &  & \hspace*{0.3cm}$d^{*}_2\mathcal{L}_1+f^{*}_1\mathcal{L}_2,$ &  & \\
        &  & \hspace*{0.3cm}$f^{*}_2\mathcal{L}_1+d^{*}_1\mathcal{L}_2\}$ &  & \\\hline
\end{tabular}

\vspace*{0.2cm}
\textbf{Aggregation Operator}
\vspace*{0.2cm}

\newcolumntype{L}[1]{>{\raggedright\let\newline\\\arraybackslash\hspace{0pt}}m{#1}}
\newcolumntype{C}[1]{>{\centering\let\newline\\\arraybackslash\hspace{0pt}}m{#1}}

\renewcommand{\tabcolsep}{1mm}
\begin{tabular}{|L{2cm}||c||C{3.22cm}|} \hline
\textbf{Operator}                 & \cellcolor{gray!25} \textbf{Approximate} & \cellcolor{gray!25} \textbf{Estimated}  \\
                                  & \cellcolor{gray!25} \textbf{Output} & \cellcolor{gray!25} \textbf{Error}  \\\hline
\textsf{Sum(T,$\ell_s,\ell_e$)}   & $\sum_{i=\ell_s}^{\ell_e}f(i)$  &    $\mathcal{L}$ \\\hline
\end{tabular}

\vspace*{0.2cm}
\textbf{Arithmetic Operators}
\vspace*{0.2cm}
\renewcommand{\tabcolsep}{1mm}
\begin{tabular}{|l||c||c|} \hline
\textbf{Operator}                 & \cellcolor{gray!25} \textbf{Approximate} & \cellcolor{gray!25} \textbf{Estimated}  \\
				                  & \cellcolor{gray!25} \textbf{Output} & \cellcolor{gray!25} \textbf{Error}  \\\hline
\textit{Agg $+$ Number}             & $\hat{Agg} + Number$            &  $\hat{\varepsilon}$ \\\hline
\textit{Agg $-$ Number}             & $\hat{Agg} - Number$            &   $\hat{\varepsilon}$ \\\hline
\textit{Agg $\times$ Number}        & $\hat{Agg} \times Number$       &   $\hat{\varepsilon}\times number$ \\\hline
\textit{Agg $\div$ Number}          & $\hat{Agg} \div Number$         &   $\hat{\varepsilon}\div number$ \\\hline\hline
\textit{$Agg_a + Agg_b$}             & $\hat{Agg}_a + \hat{Agg}_b$            &  $\hat{\varepsilon}_a+\hat{\varepsilon}_b$ \\\hline
\textit{$Agg_a - Agg_b$}             & $\hat{Agg}_a - \hat{Agg}_b$            &   $\hat{\varepsilon}_a+\hat{\varepsilon}_b$ \\\hline
\textit{$Agg_a \times Agg_b$}        & $\hat{Agg}_a \times \hat{Agg}_b$       &   $\hat{Agg}_a\hat{\varepsilon}_b+\hat{Agg}_b\hat{\varepsilon}_a+\hat{\varepsilon}_a\hat{\varepsilon}_b$ \\\hline
\textit{$Agg_a \div Agg_b$}          & $\hat{Agg}_a \div \hat{Agg}_b$         &   $\frac{\hat{Agg}_a+\hat{\varepsilon}_a}{\hat{Agg}_b-\hat{\varepsilon}_b}-\frac{\hat{Agg}_a}{\hat{Agg}_b}$ \\\hline
\end{tabular}
\caption{Formulas for estimating answer and error for each algebraic operator (single segment).}
\label{fig:error_estimators}
\end{figure}

Without going into detail into each of them, we next explain how they can be used to compute the answer and corresponding error guarantees for an entire query through an example.

\begin{figure}[h]
\centering
\includegraphics[width=0.46\textwidth]{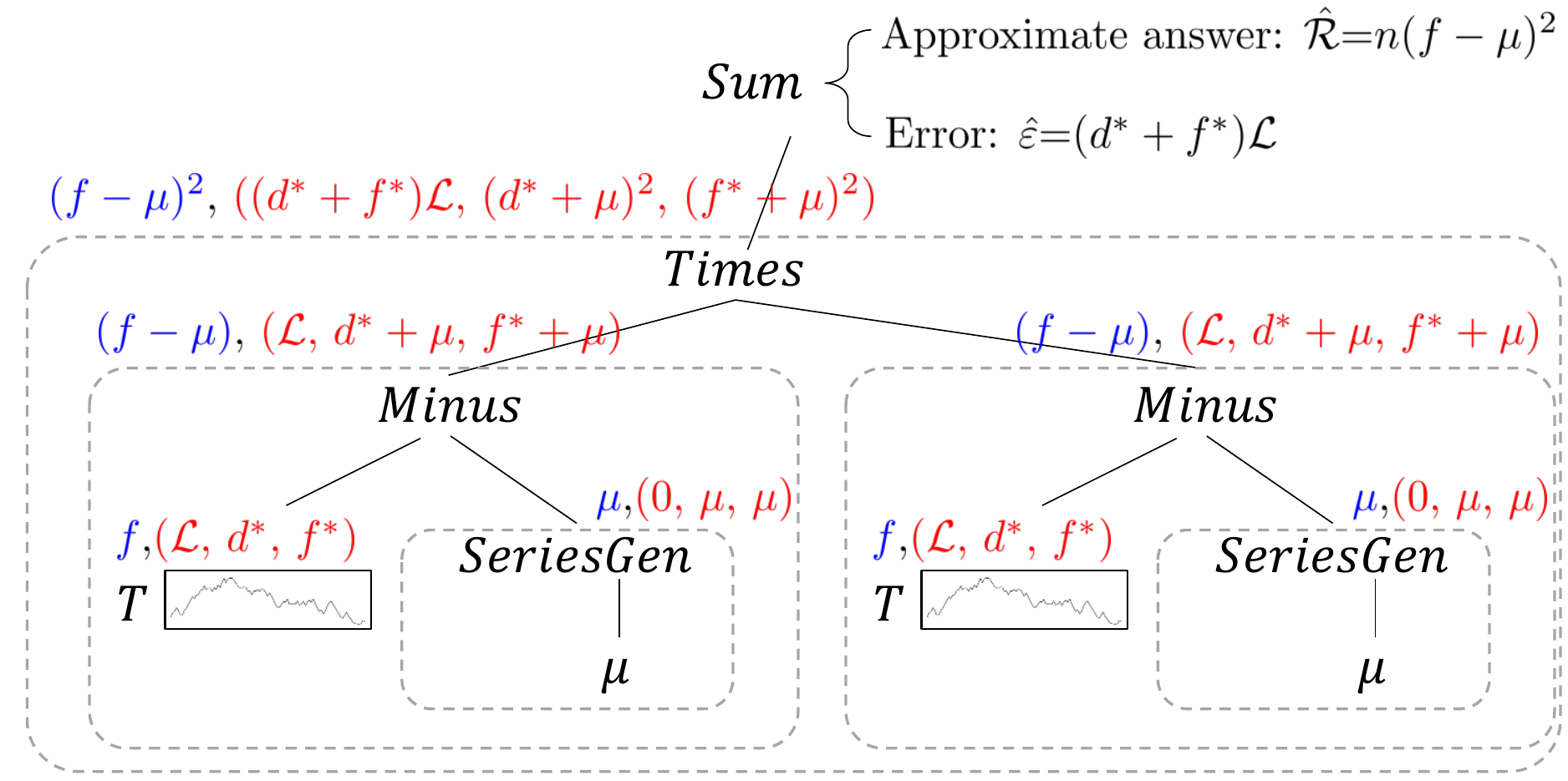}
\caption{Approximate query answer and associated error for query $Q=\textsf{Sum}(\textsf{Times}$ $(\textsf{Minus}(T,\textsf{SeriesGen}(\mu,n)),$ $\textsf{Minus}(T,\textsf{SeriesGen}(\mu,n)),1,n)$. Compression functions and error measures are shown in blue and red, respectively.}
\label{fig:error_estimator_example}
\end{figure}

\begin{example}
This example shows how to use the formulas in Figure~\ref{fig:error_estimators} to compute the approximate answer and associated error for a query computing the variance of a time series $T$ consisting of single segment $S$. For simplicity of the query expression we assume that the mean $\mu$ of $T$ is known in advance (note that even if $\mu$ was not known, the query would still be expressible in \db's query language, albeit through a longer expression). Let $f$ be the compression function and $(\mathcal{L}, d^*, f^*)$ the error measures of $S$. The query can be expressed as $Q=\textsf{Sum}(\textsf{Times}$ $(\textsf{Minus}(T,\textsf{SeriesGen}(\mu,n)),$ $\textsf{Minus}(T,\textsf{SeriesGen}(\mu,n)),1,n)$. Figure \ref{fig:error_estimator_example} shows how \db\ evaluates this query in a bottom-up fashion. It first uses the formula of the \textsf{SeriesGen} operator to compute the compression function ($f=\mu$) and error measures ($\mathcal{L}=0$, $d^{*}=\mu$, $f^{*}=\mu$) for the output of the \textsf{SeriesGen} operator. It then computes the compression function ($f-\mu$) and error measures ($\mathcal{L}$, $(d^{*}+\mu)$, $(f^{*}+\mu)$) for the output of the \textsf{Minus} operator. The computation continues in a bottom-up fashion, until \db\ computes the output of the \textsf{Sum} operator in the form of an approximate answer $\hat{\mathcal{R}}=n(f-\mu)^2$ where $n$ is the number of data points in $T$, and an estimated error $\hat{\varepsilon}=(d^*+f^*)\mathcal{L}$.
\end{example}

Importantly, the formulas shown in Figure \ref{fig:error_estimators} are guaranteed to produce the best error estimation out of any formula that uses the three error measures employed by \db\, as explained by the following theorem:

\begin{theorem}
\label{theorem:lower_bound}
The estimated errors produced through the use of the formulas shown in Figure \ref{fig:error_estimators} are the lowest among all possible error estimations produced by using the error measures described in Section \ref{sec:index}.$\Box$
\end{theorem}
The proof can be found in Appendix~\ref{appendix:lower_bound}.

\subsection{Multiple Segment Time Series}

Let us now consider the general case, where each time series $T$ contains multiple segments of varying different sizes. As a result of the varying sizes of the segments, segments of different time series may not fully align.

\begin{example}
\label{example:triggered_segment}
For instance consider the top two time series $T_{1}=(S_{1,1},S_{1,2})$ and $T_{2}=(S_{2,1},S_{2,2})$ of Figure~\ref{fig:aligned_segments} (ignore the third time series for now). Segment $S_{1,1}$ overlaps with both $S_{2,1}$ and $S_{2,2}$. Similarly, segment $S_{2,2}$ overlaps with both $S_{1,1}$ and $S_{1,2}$.
\end{example}

One may think that this can be easily solved by creating subsegments that are perfectly aligned and then using for each of them the answer and error estimation formulas of Section \ref{Sec:error_estimator_single}.

\begin{example}
Continuing our example, the two time series $T_{1}$ and $T_{2}$ can be split into the three aligned subsegments shown as the output time series $T_{3}$. Then for each of these output segments, we can compute the error based on the formulas of Section \ref{Sec:error_estimator_single}.
\end{example}

However, the problem with this approach is that the resulting error will be severely \emph{overestimated} as the error of a single segment of the original time series may be counted multiple times, as it overlaps with multiple output segments.

\begin{example}
For instance, for a query over the time series $T_{1}$ and $T_{2}$ of Figure~\ref{fig:aligned_segments},the error of $S_{2,2}$ will be double-counted, as it will be counted towards the error of the two output segments $S_{3,2}$ and $S_{3,3}$. \end{example}

To avoid this pitfall, \db\ does not estimate the error for its segment individually but instead computes the error holistically for the entire time series. Figures \ref{table:error_estimator_time_series_operator_2} and \ref{table:error_estimator_aggregation_operator_multiple} show the resulting answer and error estimation formulas for time series operators and the aggregation operator, respectively. The formulas of the arithmetic operators are omitted as they remain the same as in the single segment case, as the arithmetic operators take as input single numbers instead of time series and are thus not affected by multiple segments.

%Let $T_1=(S_{1,1},...,S_{1,p})$, $T_2=(S_{2,1},...,S_{2,q})$ be two time series with $p$ segments and $q$ segments respectively. Let $f_{i,j}$ be the compression function  and $(\mathcal{L}_{i,j}, d_{i,j}^*, f_{i,j}^*)$ be the error measures of segment $S_{i,j}$. Since segments may have different sizes, they are not perfectly aligned.  Consider a time series operator\footnote{It can be one of the \textsf{Plus}, \textsf{Minus} and \textsf{Times} operator.} that takes $T_1$ and $T_2$ as inputs and outputs a time series $T_3$=$(S_{3,1},...,S_{3,k})$ with $k$ segments $(k\geq max(p,q))$. Assume the segment $S_{3,i}$ is located in range $[i,j]$ in $T_3$, let $S_{1,u}$ be the segment in $T_1$ covers the range, and $S_{2,v}$ be the segment in $T_2$ contains the range. Then segments $S_{1,u}$ and $S_{2,v}$ are called \textit{triggered segments} of segment $S_{3,i}$, which can be marked as $(S_{1,u}, S_{2,v})$$\twoheadrightarrow$$S_{3,i}$.

%\begin{example}
%\label{example:triggered_segment}
%Consider the time series in Figure~\ref{fig:aligned_segments}  $T_{1}=(S_{1,1},S_{1,2})$ and $T_{2}=(S_{2,1},S_{2,2})$. Assume a time series operator takes $T_1$ and $T_2$ as inputs and outputs $T_3=(S_{3,1},S_{3,2},S_{3,3}$). Seen from Figure~\ref{fig:aligned_segments}, we have $(S_{1,1},S_{2,1})$$\twoheadrightarrow$$S_{3,1}$, $(S_{1,1},S_{2,2})$$\twoheadrightarrow$$S_{3,2}$, and $(S_{1,2},S_{2,2})$$\twoheadrightarrow$$S_{3,3}$.
%\end{example}

\begin{figure}[t]
\centering
\includegraphics[width=0.35\textwidth]{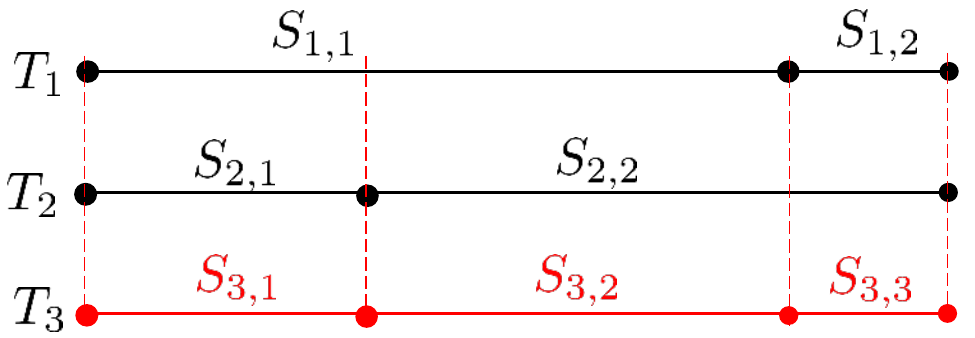}
\caption{Example of aligned time series segments. The new generated time series $T_3$ is shown in red color.}
\label{fig:aligned_segments}
\end{figure}

%Next we describe the details of (i) how to propagate compression function and error measures for time series operators, (ii) how to get approximate answers annotated with errors for the aggregation operator, and (iii) how to handle the arithmetic operators.

\begin{figure*}[t]
\small
\raggedright
\textbf{Time Series Operators}
\vspace*{0.2cm}
\renewcommand{\tabcolsep}{0.2mm}
\begin{tabular}{|l||A||c|c|c|} \hline
\textbf{Operator}                 & \cellcolor{gray!25} \textbf{Comp. func.} &\multicolumn{3}{c|}{\cellcolor{gray!25} \textbf{Output Error Measures}} \\\cline{2-5}
                                  & \cellcolor{gray!13} f & \cellcolor{gray!13} $\mathcal{L}$  &   \cellcolor{gray!13}  $d^*$     &  \cellcolor{gray!13} $f^*$\\\hline
\textit{SeriesGen($\upsilon, n$)} & $\upsilon$ &$0$            &    $\upsilon$& $\upsilon$ \\\hline
\textit{Plus($T_a,T_b$)}          & $\{(f_{c,1},...,f_{c,k})| f_{c,i}=f_{a,u}+f_{b,v}~i\in[1,k]\}$ & $\sum_{i=1}^{p}\mathcal{L}_{a,i}+\sum_{j=1}^{q}\mathcal{L}_{b,j}$ & $max\{d_{c,i}| d_{c,i}=d_{a,u}+d_{b,v}~i\in[1,k]\}$ & $max\{f_{c,i}| f_{c,i}=f_{a,u}+f_{b,v}~i\in[1,k]\}$\\\hline
\textit{Minus($T_a,T_b$)}         & $\{(f_{c,1},...,f_{c,k})| f_{c,i}=f_{a,u}-f_{b,v} ~i\in[1,k]\}$ & $\sum_{i=1}^{p}\mathcal{L}_{a,i}+\sum_{j=1}^{q}\mathcal{L}_{b,j}$ & $max\{d_{c,i}| d_{c,i}=d_{a,u}+d_{b,v}~ i\in[1,k]\}$ & $max\{f_{c,i}| f_{c,i}=f_{a,u}+f_{b,v}~ i\in[1,k]\}$\\\hline
\textit{Times($T_a,T_b$)}         & $\{(f_{c,1},...,f_{c,k})| f_{c,i}=f_{a,u}\times f_{b,v} \forall i\in[1,k]\}$ & $\mathcal{L}_{T_c}$ & $max\{d_{c,i}| d_{c,i}=d_{a,u}\times d_{b,v}~ i\in[1,k]\}$ & $max\{f_{c,i}| f_{c,i}=f_{a,u}\times f_{b,v}~ i\in[1,k]\}$\\\hline
\end{tabular}
\caption{Formulas for estimating answer and error for time series operators (multiple segments). For each output time series segment $S_{c,i}$, let $S_{a,u}$ and $S_{b,v}$ be the input segments that overlap with $S_{c,i}$.}
\label{table:error_estimator_time_series_operator_2}
\end{figure*}

\begin{figure}
\small
\raggedright
\textbf{Aggregation Operator}
\vspace*{0.2cm}
\renewcommand{\tabcolsep}{3.9mm}
\begin{tabular}{|l||c||c|} \hline
\textbf{Operator}                 & \cellcolor{gray!25} \textbf{Approximate} & \cellcolor{gray!25} \textbf{Estimated}  \\
                                  & \cellcolor{gray!25} \textbf{Output} & \cellcolor{gray!25} \textbf{Error}  \\\hline
\textit{Sum(T,$\ell_s,\ell_e$)}   & $\sum_{i=u}^{v}\sum_{j=1}^{|S_i|}f_i(j)$  &    $\sum_{i=u}^{v}\mathcal{L}_i$ \\\hline
\end{tabular}
\caption{Formulas for estimating answer and error for the aggregation operator (multiple segments).}
\label{table:error_estimator_aggregation_operator_multiple}
\end{figure}

%\noindent\textbf{Arithmetic operator.} Surprisingly, the expressions of computing the approximate answers and the errors for the arithmetic operator are exactly the same to those in Table~\ref{table:error_estimator_arithmetic_operator}, which indicates that the arithmetic operators are independent to the number of segments. This is because all the arithmetic operators are operating numeric values instead of time series.

%For the \textit{Times($T_1,T_2,\ell_1,\ell_2,n$)} operator, the estimated error estimator is $\hat{\varepsilon}=\sum\limits_{i=1}^{k_1}min(C_1,...,C_{2^{m_i}})$, where $m_i$ is the number of segments in $T_2$ overlapping with the segment $S_{i}^{(1)}$ in $T_1$.
%
%\[C=\overline{d} \cdot\mathcal{L}^{(1)}+\hat{d}^{(1)*}\cdot\underbrace{(\mathcal{L}_u^{(2)},...,\mathcal{L}_{u+a-1}^{(2)})}_{a}+d^{(1)*}\cdot\underbrace{(\mathcal{L}_{v}^{(2)},...,\mathcal{L}_{v+(m-a)-1}^{(2)})}_{m-a}\]
%where $\overline{d}$ is the maximal value among  $m$ values chosen from the following $2m$ values $(d_1^{(2)*},...,d_{m}^{(2)*},\hat{d}_1^{(2)*},...,\hat{d}_{m}^{(2)*})$. That is
%
%\[\overline{d} = max\{\underbrace{d_u^{(2)*},...,d_{u+a-1}^{(2)*}}_{a},\underbrace{\hat{d}_{v}^{(2)*},...,\hat{d}_{v+(m-a)-1}^{(2)*}}_{m-a}\}\]

\section{Navigating the \treeCapAll}
\label{sec:QueryProcessing}

So far we have seen how \db\ computes the approximate answer to a query and its associated error, assuming that the segments that are used for query processing have already been selected. In this Section, we explain how this selection is performed. In particular, we show how \db\ navigates the \tree s of the time series involved in the query to continuously compute better estimations of the query answer under the given error or time budget is satisfied.\\  

%In this section, we describe how to return approximate answers with estimated errors by using the hierarchical tree structures. The estimated errors are guaranteed to be no greater than the error budgets given by users.

\noindent\textbf{Query Processing Algorithm.} Let $T_1,...,T_m$ be a set of time series and $\mathcal{I}_{1},...,\mathcal{I}_{m}$ the respective \tree s. Let also $Q$ be a query over $T_1,...,T_m$ and $\varepsilon_{max}$/$t_{max}$ an error/time budget, respectively. To answer $Q$ under the given budget, \db\ first starts from the roots of $\mathcal{I}_{1},...,\mathcal{I}_{m}$ and uses them to compute the approximate query answer $\hat{\mathds{R}}$ and corresponding error $\hat{\varepsilon}$ using the formulas presented in Section \ref{Sec:error_estimator}. If the estimated error is greater than the error budget (i.e., if $\hat{\varepsilon}\geq\varepsilon_{max}$) or if the elapsed time is smaller than the allowed time budget, \db\ chooses one of the tree nodes used above, replaces it with its children and repeats the above procedure using the newly selected nodes until the given error/time budget is reached. What is important is the criterion that is used to choose the node that is replaced at each step by its children. In general, \db\ will have to select between several nodes, as it will be exploring in which \tree\ and moreover in which part of the selected \tree\ it pays off to navigate further down. Since \db\ aims to reduce the estimated error as much as possible, at each step it greedily chooses the node whose replacement by its children leads to the biggest reduction in the estimated error. The resulting procedure is shown as Algorithm~\ref{alg:queryprocessing}~\footnote{Note that the algorithm is shown for both error and time budget case. In contrast to the case when a time budget is provided, in which the algorithm has to always keep a computed estimated answer $\hat{\mathds{R}}$ to return it when the time budget runs out, in the case of the error budget this is not required. Thus, in the latter case, it suffices to compute $\hat{\mathds{R}}$ only at the very last step of the algorithm, thus avoiding its iterative computation during the while loop.}.

\begin{algorithm}[h!]\small
\KwIn{\treeCap s $\mathcal{I}_{1},...,\mathcal{I}_{m}$, query $Q$, error budget $\varepsilon_{max}$ or time budget $t_{max}$}
\KwOut{Approximate answer $\hat{\mathds{R}}$ and error $\hat{\varepsilon}$}
Access the roots of $\mathcal{I}_{1},...,\mathcal{I}_{m}$;\\
Compute $\hat{\mathds{R}}$ and $\hat{\varepsilon}$ by using the compression functions and error measures of the currently accessed nodes (see Section~\ref{Sec:error_estimator} for details);\\
\While{$\hat{\varepsilon}>\varepsilon_{max}$ or $\text{elapsed time} < t_{max}$}
{
  Choose a node maximizing the error reduction;\\
 %\tcp{Let $N_l$ and $N_r$ be two children nodes of the node $N$}
  Update the current answer $\hat{\mathds{R}}$ and error $\hat{\varepsilon}$ using the compression functions and error measures of the currently accessed nodes; \label{alg:line5}

%  Use the error measures in $N$, $N_l$ and $N_r$  to update the estimated error $\hat{\varepsilon}$ based on the expressions in Table~\ref{table:error_estimator_update};

  %  Choose a node $N$ satisfying the condition defined in XXX;\\
   % Using the error parameters of $N_l$ and $N_r$  to incrementally update $\hat{\varepsilon}$ (Section~\ref{sec:incremental_update});
}
%Compute the approximate answer  $\hat{\mathds{R}}$ of $Q$ using the compression functions of the currently accessed nodes;\\
\textbf{Return} $(\hat{\mathds{R}}, \hat{\varepsilon})$;
\caption{\db{} Query Processing}
\label{alg:queryprocessing}
\end{algorithm}

\smallskip
\noindent\textbf{Algorithm Optimality.} Given its greedy nature, one may wonder whether the query processing algorithm is optimal. To answer this question, we have to first define optimality. Since the aim of the query processing algorithm is to produce the lowest possible error in the fastest possible time (which can be approximated by the number of nodes that are accessed), we say that an algorithm is optimal if for every possible query, set of \tree s, and error budget $\varepsilon_{max}$ it answers the query under the given budget accessing the lowest number of nodes than any other possible algorithm. Since a comparison of any possible algorithm is hard, we also restrict our attention to deterministic algorithms that access the \tree s in a top-down fashion (i.e, to access a node $N$ all its ancestor nodes should also be accessed). We denote this class of algorithms as $\mathbb{A}$. It turns out that no algorithm in $\mathbb{A}$ can be optimal as the following theorem states:

\begin{theorem}
\label{theorem:no_optimal}
There is no optimal algorithm in $\mathbb{A}$.
\end{theorem}

\begin{proof}
Consider the following \tree s of two time series $T_1$ and $T_2$. The \tree\ of $T_1$ is shown in Figure~\ref{fig:no_optimal} and the \tree\ of $T_2$ is a tree containing a single node. Now consider a query $Q$ over these two time series and an error budget $\varepsilon=h-1$ where $h > 1$ is the height of the $T_1$'s tree. Assume that the query error using the tree roots is $\varepsilon_{root}=2h$. Also assume that whenever the query processing algorithm replaces a node by its children, the  error for the query is reduced by $\frac{1}{2^h}$ with the exception of the shaded node, which, when replaced by its children, leads to an error reduction of $h+1$. This means that the query processing algorithm can only terminate after accessing the children of the shaded node, as the query error in that case will be at most $2h-(h+1)=h-1$. Otherwise, the error estimated by the algorithm will be at least $2h-2^h(\frac{1}{2^h})=2h-1>h-1$, which exceeds the error budget and thus does not allow the algorithm to terminate. Since the shaded node can be placed at an arbitrary position in the tree, for every given deterministic algorithm, we can place the shaded node in the tree, so that the algorithm accesses the children of the shaded node only after it has accessed all the other nodes in the tree. However, this is suboptimal, as there is a way to access the children of the shaded node with fewer node accesses (i.e., by following the path from the root to the shaded node). Therefore, no algorithm in $\mathbb{A}$ is optimal.
\end{proof}
\begin{figure}[h]
\centering
\includegraphics[width=0.35\textwidth]{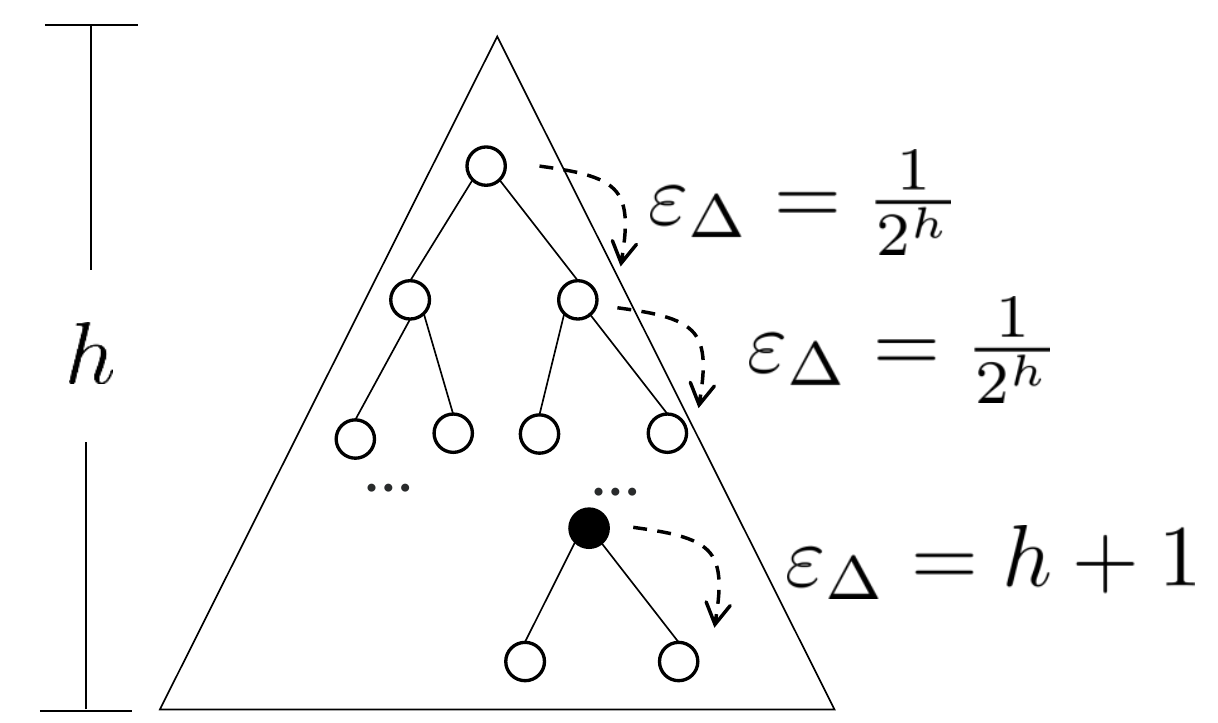}
\caption{\treeCap\ for Theorem~\ref{theorem:no_optimal}.}
\label{fig:no_optimal}
\end{figure}

As a result of the above theorem, \db's query processing algorithm cannot be optimal in general. However, we can show that it is optimal for \tree s that exhibit the following property: For every pair of nodes $N$ and $N'$ of the \tree, such that $N'$ is a descendant of $N$, the error reduction $\varepsilon_{\Delta}(N)$ achieved by replacing $N$ with its children is greater or equal to the error reduction $\varepsilon_{\Delta}(N')$ achieved by replacing $N'$ with its children. Such a tree is called \emph{fine-error-reduction} tree and intuitively it guarantees that any node leads to a greater or equal error reduction than any of its descendants. If all trees satisfy the above property, \db's query processing algorithm is optimal:   

\begin{theorem}
\label{theorem:greedy_optimal}
In the presence of \tree s that are \textit{fine-error-reduction} trees, \db's query processing algorithm is optimal.
\end{theorem}

%Jumping from a node $N$ to its children nodes $N_l$ and $N_r$, an algorithm can obtain an error reduction, called $\varepsilon_{\Delta}(N)$. For any node $N$, if $\nexists$ a node $N'$ in the subtree of $N$ such that $\varepsilon_{\Delta}(N')>\varepsilon_{\Delta}(N)$, then the tree structure is a  \textit{fine-error-reduction} tree. If all the trees are  \textit{fine-error-reduction} trees, then \db{} is optimal. This is because a node is not chosen by \db\ to access its children node at this stage, then all the nodes in its subtree cannot get better error reduction.

\smallskip
\begin{table}[h!]
\small
\centering
\renewcommand{\tabcolsep}{1mm}
\begin{tabular}{|l|B|} \hline
\textbf{Operator} & \textbf{Incremental Error Update} \\\hline
\textit{Plus($T_a,T_b$)}   &  $\hat{\varepsilon'}=\hat{\varepsilon}-(\mathcal{L}_a-(\mathcal{L}_{a.1}+\mathcal{L}_{a.2}))$\\\hline
\textit{Minus($T_a,T_b$)}  &  $\hat{\varepsilon'}=\hat{\varepsilon}-(\mathcal{L}_a-(\mathcal{L}_{a.1}+\mathcal{L}_{a.2}))$\\\hline
\textit{Times($T_a,T_b$)}  &  $\hat{\varepsilon'}=\hat{\varepsilon}-(max(p_{b,1},...,p_{b,k})\mathcal{L}_a - max(p_{b,1},...,p_{b,i})\mathcal{L}_{a.1} + max(p_{b,i},...,p_{b,k})\mathcal{L}_{a.2})$\\\hline
%Divide operator \textit{Divide($T_1,T_2,\ell_1,\ell_2,n$)} &  $\mathcal{L}_1+\mathcal{L}_2$\\\hline
\end{tabular}
\caption{Incremental update of estimated errors for time series operators. $p_{b,i}\in\{d^*_{b,i}, f^*_{b,i}\}$.}
\label{table:error_estimator_update}
\end{table}
\noindent\textbf{Incremental Error Update.}
Having proven the optimality of the algorithm for fine-error-reduction trees, we will next discuss an optimization that can be employed to speedup the algorithm. By studying the algorithm, it is easy to observe that as the algorithm moves from a set $\mathcal{N} = \{N_1, \ldots, N_n\}$ of nodes to a set $\mathcal{N'} = \{N_1$, $\ldots$, $N_{a-1}$, $N_{a.1}$, $N_{a.2}$, $N_{a+1}$, $\ldots$, $N_{n}\}$ of nodes (by replacing node $N_a$ by its children $N_{a.1}$ and $N_{a.2}$), it recomputes the error using all nodes in $\mathcal{N'}$, although only the two nodes $N_{a.1}$ and $N_{a.2}$ have changed from the previous node set $\mathcal{N}$.

This observation led to the incremental error update optimization of \db's query processing algorithm described next. Instead of recomputing from scratch the error of $N'$ using all nodes, \db\ incrementally updates the error of $N$ by using only the error measures of the newly replaced node $N_{a}$ and the newly inserted nodes $N_{a.1}$ and $N_{a.2}$. Let $(\mathcal{L}_a, d^{*}_a, f^{*}_a)$, $(\mathcal{L}_{a.1}, d^{*}_{a.1}, f^{*}_{a.1})$, and $(\mathcal{L}_{a.2}, d^{*}_{a.2}, f^{*}_{a.2})$ be the error measures of nodes $N_a$, $N_{a.1}$, and $N_{a.2}$, respectively. Assume that the segments $S_{b,1},...,S_{b,k}$ overlap with the segment of node $N_{a}$, the segments $S_{b,1},...,S_{b,i}$ $(i\leq k)$ overlap with the segment of node $N_{a.1}$, and the segments $S_{b,i},...,S_{b,k}$ overlap with the segment of node $N_{a.2}$. Then the estimated error $\hat{\varepsilon'}$ using nodes $N_{a.1}$ and $N_{a.2}$ can be incrementally computed from the error $\hat{\varepsilon}$ using node $N_a$ through the incremental error update formulas shown in Table~\ref{table:error_estimator_update}\footnote{The \textit{SeriesGen} operator is omitted, since its input is not a time series and as a result there is no \tree\ associated with its input.}.\\

\noindent\textbf{Probabilistic Extension.}
While \db\ provides deterministic error guarantees, which as we discussed above are in many cases required, it is interesting to note that it can be easily extended to provide probabilistic error guarantees if needed. Most importantly this can be done simply by changing the error measures computed for each segment from $(\mathcal{L}, d^*, f^*)$ to $(\sigma_{\varepsilon}, \varepsilon^*, f^*)$, where $\sigma_{\varepsilon}$ is the variance of $d_i-f(i)$, and $\varepsilon^*$ is the maximal absolute value of $d_i-f(i)$. Then we can employ the Central Limit Theorem (CLT)~\cite{dudley1999uniform} to bound the accurate error $\varepsilon$ by $Pr(\varepsilon\leq\hat{\varepsilon})\geq 1-\alpha$, where $\alpha$ can be adjusted by the users to get different confidence levels. It is interesting that the rest of the system, including the hierarchical structure of the \tree\ and the tree navigation algorithm employed at query processing time do not need to be modified. In our future work we plan to further explore this probabilistic extension and compare it to existing approximate query answering techniques with probabilistic guarantees.
\section{Experimental Evaluation}
\label{experiments}

To evaluate \db's performance and verify our hypothesis that \db\ is able to provide significant savings in the query processing of sensor data, we are conducting experiments on real sensor data. We present here early data points that we have discovered.\\

\noindent\textbf{Datasets.} For our preliminary experiments, we used two real sensor datasets:

\begin{enumerate}
  \item \emph{Intel Lab Data (ILD)}\footnote{\url{http://db.lcs.mit.edu/labdata/labdata.html}}. Smart home data (humidity and temperature) collected at 31-second intervals from 54 sensors deployed at the Intel Berkeley Research Lab between February 28th and April 5th, 2004. The dataset contains about 2.3 million tuples (i.e., 4.6 million sensor readings in total).
  \item \emph{EPA Air Quality Data (AIR)}\footnote{\url{https://www.epa.gov/outdoor-air-quality-data}}. Air quality data collected at hourly intervals from about 1000 sensors from January 1st 2000 to April 1st 2016. The dataset contains about 133 million tuples (i.e., 266 million sensor readings in total).
      %\yannisk{Does the dataset contain only these two values or are these just examples? Chunbin: we exclude the other attributes, only store these two.}
\end{enumerate}

From each dataset we extracted multiple time series, each corresponding to a single attribute of the dataset; Humidity and Temperature for ILD and Ozone and SO$_2$ for AIR. We then used \db\ to create the corresponding segment tree for each time series and to answer queries over them.\\

%\yannisk{Chunbin, please fill in the information above or revise if you created multiple time series from each dataset. Done}\\

\noindent\textbf{Experimental platform}. All experiments were performed on a computer with a 4th generation Intel i7-4770 processor ($4\times 32$ KB L1 data cache, $4 \times 256$ KB L2 cache, $8$ MB shared L3 cache, $4$ physical cores, $3.6$ GHz) and $16$ GB RAM,
%and a Seagate ST2000DM001-1CH1 hard drive,
running Ubuntu 14.04.1. All the algorithms were implemented in C++ and compiled with g++ 4.8.4, using -O3 optimization. All data was stored in main memory.

\subsection{Experimental Results}

In our preliminary evaluation, we measured two quantities: First, the size of the segment tree created by \db, since this segment tree is stored in main memory, and second, the query processing performance of \db\ compared to a system that answers queries using the entirety of the raw sensor data. In our future work, we will be conducting a more thorough evaluation of the system. We next present our preliminary results:

\begin{table}[htp]
\small
\centering
\renewcommand{\tabcolsep}{0.7mm}
\begin{tabular}{|c|c|c|c|c|} \hline
\textbf{Dataset}&\textbf{$\#$ Tuples}&\textbf{Raw Data} & \multicolumn{2}{c|}{\textbf{Segment Tree}}\\
                &                    &              & \textbf{(0-degree)}& \textbf{(1-degree)} \\\hline
          ILD   & 2,313,153	        & 35.29 MB       & 0.14 MB & 0.67 MB\\\hline
          AIR   & 133,075,510       & 1.98 GB        & 4.37 MB & 8.11 MB\\\hline
\end{tabular}
\caption{Raw data and segment tree sizes.}
\label{table:size}
\end{table}

\noindent\textbf{Segment tree size.}
Table~\ref{table:size} shows the size of the raw data and the combined size of the segment trees built for all the time series extracted from the ILD and AIR datasets.\footnote{To make a fair comparison, the raw data size refers only to the combined size of the attributes used in the time series and does not include other attributes that exist in the original dataset (such as location codes etc).} We experimented with two different compression functions, resulting in different segment tree sizes; a 0-degree polynomial (corresponding to the Piecewise Aggregate Approximation \cite{KeoghCPM01}, where each value within a segment is approximated through the average of the values in the segment) and a 1-degree polynomial (corresponding to the Piecewise Linear Approximation \cite{keogh1997fast}, where each segment is approximated through a line). 
%\yannisk{Chunbin, please read the above to check whether it correctly describes the compression functions used in the experiments. Chunbin: yes, correct.}
As shown, the segment trees are significantly smaller than the  raw sensor data (about  $0.40\% - 1.90\%$ and $0.22\% - 0.40\%$ smaller for the ILD and AIR datasets, respectively). As a result, the segment trees of the time series can be easily kept in main memory, even when the system stores a large number of time series.\\
%\yannisk{Are the sizes of the 0-degree and 1-degree segment trees in the table reversed or is the 1-degree tree indeed smaller than the 0-degree tree? Oh, it is reversed. I have fixed it.} 
%\yannisk{Does the data size and number of tuples refer to the total size of the dataset (for all attributes) or only for the attribute on which we built the index? If it is the latter, explain that in the text, and change the numbers on the table. Also coordinate with the size mentioned in the dataset list at the beginning of the section. Chunbin: the size the data that only contains the attributes that we built the index.}

\begin{figure}[h]
\centering
\begin{minipage}{.46\textwidth}
\begin{tabular}{cc}
\includegraphics[width=.47\textwidth]{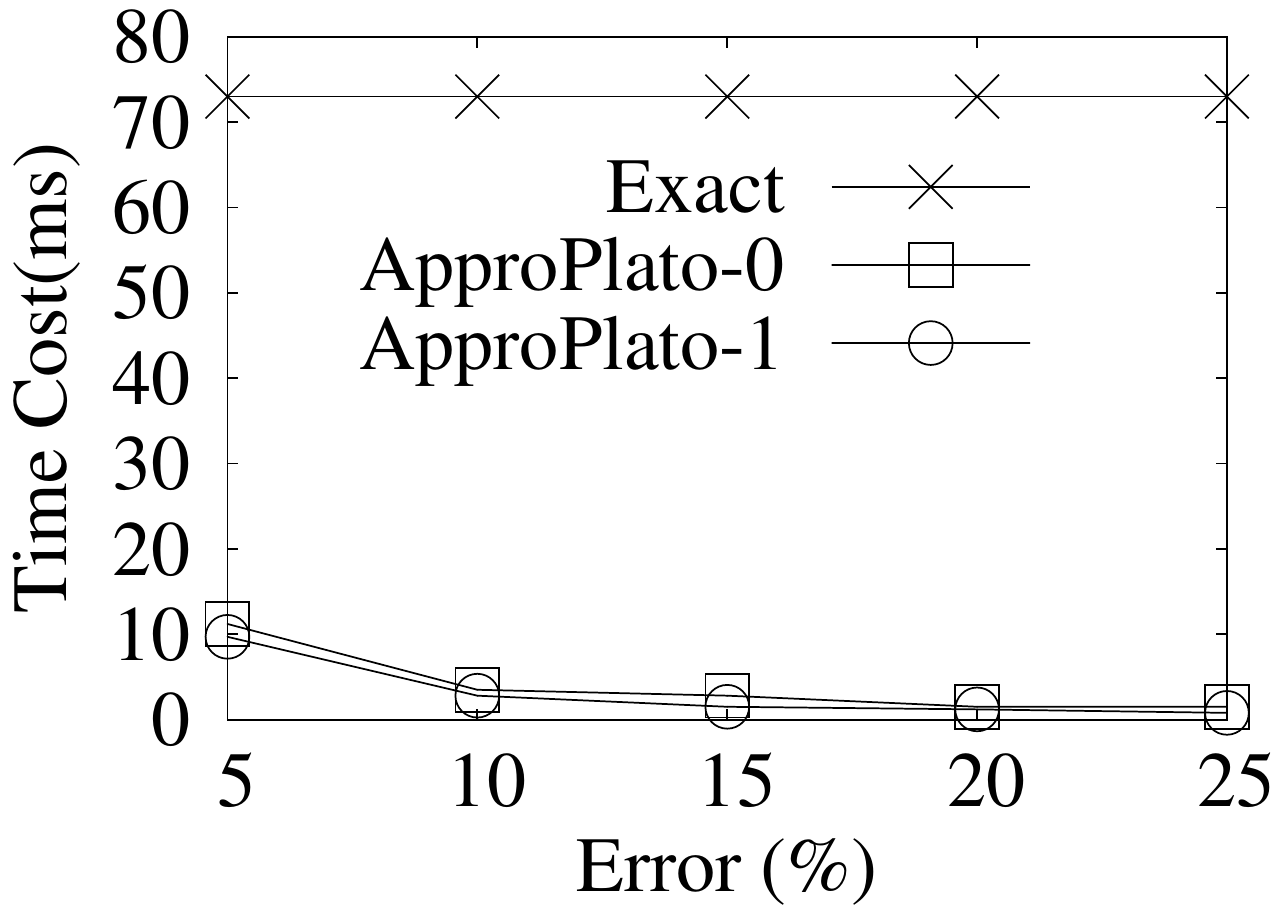}&
\includegraphics[width=.47\textwidth]{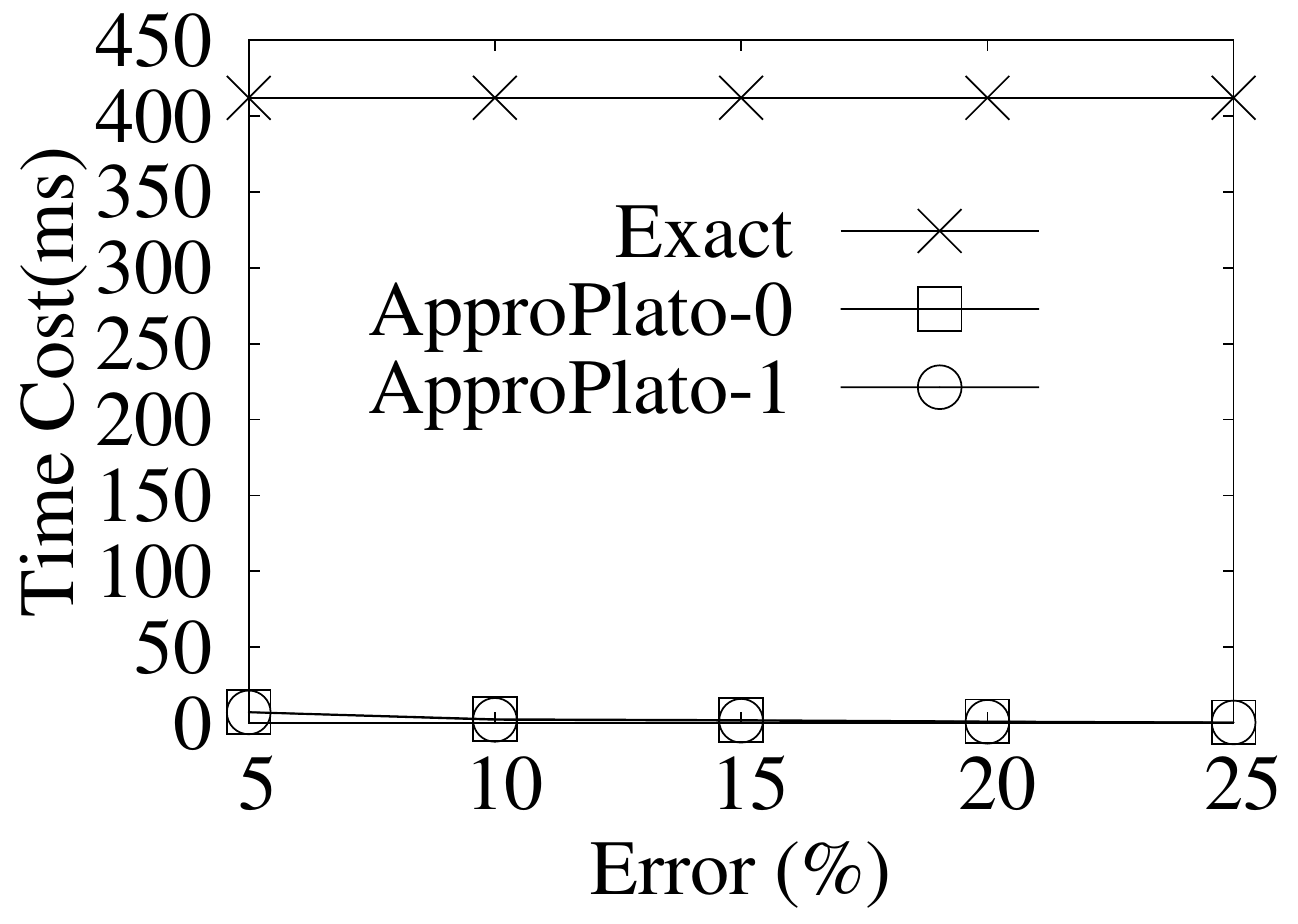}\\
\small{(a) ILD} & \small{(b) AIR}
\end{tabular}
\caption{Query processing performance for correlation query (time shown in ms).}
\label{fig:time}
\end{minipage}
\end{figure}

\noindent\textbf{Query processing performance.}
We next compared the query processing performance of \db\ against a baseline, which is a custom in-memory algorithm that computes the exact answer of the queries using the raw data. To compare the systems, we measured the time required to process a correlation query between two time series (i.e., correlation(Humidity, Temperature) in ILD and correlation(Ozone and SO$_2$) in AIR)) with a varying error budget (ranging from 5\% to 25\%). Figure~\ref{fig:time} shows the resulting times for each of the two datasets. Each graph depicts the performance of three systems; \emph{Exact}, which is the baseline method of answering queries over the raw data, and \emph{\db-0}, \emph{\db-1}, which are instances of \db\ using the 0-degree and 1-degree polynomial compression functions, as explained above.

%\yannisk{How was the baseline method implemented? Is it a custom in-memory algorithm, a DBMS, or something else?} \yannisk{Which were the two time series used in the experiments? Chunbin: fixed.}

By studying Figure~\ref{fig:time}, we can make the following observations.
\begin{itemize}
  \item Both instances of \emph{\db} outperform \emph{Exact} by one to three orders of magnitude, depending on the provided error budget.
  \item In contrast to \emph{Exact} which always uses the entire raw dataset to compute exact query answers, \emph{\db} allows the user to select the appropriate tradeoff between time spent in query processing and resulting error by specifying the desired error budget. The system adapts to the budget by providing faster responses as the allowed error budget increases;
  \item Notably, \emph{\db} remains significantly faster than \emph{Exact} even for small error budgets. In particular, \db{} is over $9\times$ and $37\times$ faster than \emph{Exact} when the error is 5\% in ILD and AIR respectively.
% Exact numbers: $9.11\times$ and $37.36\times$ 
%  \yannisk{Chunbin, can you add the numbers of the speedup for the 5\% error case? Done}
\end{itemize}

In summary, our preliminary results show that \db{} shows significant potential for speeding up query processing of ad hoc queries over large amounts of sensor data, as it outperforms exact query processing algorithms in many cases by several orders of magnitude. Moreover, it can provide such speedups, while providing deterministic error guarantees, in contrast to existing sampling-based approximate query answering approaches that provide only probabilistic guarantees, which may not hold in practice. Despite the difference in guarantees, in our future work we will be conducting a more thorough evaluation of the system comparing it also against sampling-based systems.

%\yannisk{We need a paragraph concluding the section, explaining that the preliminary results show that this is a promising approach.}

\section{Related Work}

%\yannisk{Chunbin, please read section below and address the comments. While reading, please make sure that everything we state is true (especially the absolute comments of this being the first work that...). Otherwise, please leave a comment.}

Approximate query answering has been the focus on an extensive body of work, which we will summarize next. However, to the best of our knowledge, this is the first work that provides deterministic guarantees for aggregation queries over multiple time series.\\

\noindent\textbf{Approximate query answering with probabilistic error guarantees.} Most of the existing work on approximate query processing has focused on using \emph{sampling} to compute approximate query answers by appropriately evaluating the queries on small samples of the data~\cite{JermaineAPD07,AgarwalMPMMS13,WuOT10,BabcockDM04,PansareBJC11,PansareBJC11}. Such approaches typically leverage statistical inequalities and the central limit theorem to compute the confidence interval or variance of the computed approximate answer. As a result, their error guarantees are probabilistic. While probabilistic guarantees are often sufficient, there are not suitable for scenarios where one wants to be certain that the answer will fall within a certain interval~\footnote{Note that as discussed in Section \ref{sec:QueryProcessing}, \db\ can also be extended to provide probabilistic guarantees when deterministic guarantees are not required, simply by modifying the error measures computed for each segment.}.

A special form of sampling-based methods are \emph{online aggregation} approaches, which provide a continuously improving query answer, allowing users to stop the query evaluation when they are satisfied with the resulting error  ~\cite{hellerstein1997online,CondieCAHES10,PansareBJC11}. With its hierarchical \tree, \db\ can support the online aggregation paradigm, while providing deterministic error guarantees.\\

\noindent\textbf{Approximate query answering with deterministic error guarantees.} Approximately answering queries while providing deterministic error guarantees has so far received only very limited attention~\cite{PottiP15,LazaridisM01,PoosalaIHS96}. Existing work in the area has focused on simple aggregation queries that involve a single relational table. In contrast, \db\ provides deterministic error guarantees on queries that may involve multiple time series (each of which can be though of as a single relational table), enabling the evaluation of many common statistics that span tables, such as correlation, cross-correlation and others.\\

\noindent\textbf{Approximate query answering over sensor data.} Moreover, \db\ is one of the first approximate query answering systems that leverage the fact that sensor data are not random but follow a usually smooth underlying phenomenon. The majority of existing works on approximate query answering looked at general relational data. Moreover, the ones that studied approximate query processing for sensor data, focused on the networking aspect of the problem, studying how aggregate queries can be efficiently evaluated in a distributed sensor network~\cite{madden2002tag,considine2004approximate,ConsidineLKB04}.
While these works focused on the networking aspect of sensor data, our work focuses on the continuous nature of the sensor data, which it leverages to accelerate query processing even in a single machine scenario, where historical sensor data already accumulated on the machine have to be analyzed.\\

\noindent\textbf{Data summarizations.} Last but not least, there has been extensive work on creating summarizations of sensor data. Work in this area has come mostly from two different communities; from the database community~\cite{IoannidisP95,PoosalaIHS96,PapapetrouGD12,Ting16}
%\yannisk{Chunbin, please add citations. Done}
and the signal processing community~\cite{KeoghCPM01,keogh2001locally,keogh1997fast,ChanF99,FaloutsosRM94,FaloutsosRM94}. %\yannisk{Chunbin, please add citations. Done}.

The database community has mostly focused on creating summarizations (also referred to as synopses or sketches) that can be used to answer specific queries. These include among others histograms \cite{IoannidisP95,PoosalaIHS96,GibbonsMP97,PoosalaGI99} (e.g., EquiWidth and EquiDepth histograms \cite{ShapiroC84}, V-Optimal histograms~\cite{IoannidisP95}, Hierarchical Model Fitting (HMF) histograms~\cite{WangS08}, and Compact Hierarchical Histograms (CHH) \cite{ReissGH06}), as well as sampling methods~\cite{HaasS92,ChenY17}, used among other for cardinality estimation~\cite{IoannidisP95} and selectivity estimation~\cite{PoosalaIHS96}.
%
%Several types of histograms have been proposed and evaluated experimentally in terms of their accuracy, including EquiWidth and EquiHeight [Koo80,SC84], MaxDiff, V-Optimal histograms [IP95, PIHS96]. A formal taxonomy of histograms was proposed in [PIHS96]. The V-Optimal histograms have been shown to minimize the average error for several selectivity estimation problems 
%
In contrast to such special-purpose approaches, \db\ supports a large class of queries over arbitrary sensor data.

The signal processing community on the other hand, produced a variety of methods that can be used to compress time series data. These include among others the Piecewise Aggregate Approximation (PAA)~\cite{KeoghCPM01}, the  Adaptive Piecewise Constant Approximation (APCA)~\cite{keogh2001locally}, the Piecewise Linear Representation (PLR)~\cite{keogh1997fast}, the Discrete Wavelet Transform (DWT)~\cite{ChanF99}, and the Discrete Fourier Transform (DFT)~\cite{FaloutsosRM94}. However, it has not been concerned on how such compression techniques can be used to answer general queries. \db's modular architecture allows the easy incorporation of such techniques as compression functions, that are then automatically leveraged by the system to enable approximate answering of a large number of queries with deterministic error guarantees.

\section{Conclusion}
In this paper, we proposed the \db\ system that allows users the efficient computation of approximate query answers to queries over sensor data. By utilizing the novel \tree\ data structure, \db\ creates at data import time a set of hierarchical summarizations of each time series, which are used at query processing time to not only enable the efficient processing of queries over multiple time series with varying error/time budgets but to also provide error guarantees that are deterministic and are therefore guaranteed to hold, in contrast to the multitude of existing approaches that only provide probabilistic error guarantees. Our preliminary results show that the system can in real use cases lead to several order of magnitude improvements over systems that access the entire dataset to provide exact query answers. In our future work, we plan to perform a thorough experimental evaluation of the system, in order to both study the behavior of the system in different datasets and query workloads, as well as to compare it against systems that provide probabilistic error guarantees.

{\small
\bibliographystyle{abbrv}
\bibliography{References}

\begin{thebibliography}{10}

\bibitem{AgarwalMPMMS13}
S.~Agarwal, B.~Mozafari, A.~Panda, H.~Milner, S.~Madden, and I.~Stoica.
\newblock Blinkdb: queries with bounded errors and bounded response times on
  very large data.
\newblock In {\em EuroSys}, pages 29--42, 2013.

\bibitem{BabcockDM04}
B.~Babcock, M.~Datar, and R.~Motwani.
\newblock Load shedding for aggregation queries over data streams.
\newblock In {\em ICDE}, pages 350--361, 2004.

\bibitem{bellman1961approximation}
R.~Bellman.
\newblock On the approximation of curves by line segments using dynamic
  programming.
\newblock {\em Communications of the ACM}, 4(6):284, 1961.

\bibitem{NgC04}
Y.~Cai and R.~T. Ng.
\newblock Indexing spatio-temporal trajectories with chebyshev polynomials.
\newblock In {\em SIGMOD}, pages 599--610, 2004.

\bibitem{ChanF99}
K.~Chan and A.~W. Fu.
\newblock Efficient time series matching by wavelets.
\newblock In {\em ICDE}, pages 126--133, 1999.

\bibitem{ChenY17}
Y.~Chen and K.~Yi.
\newblock Two-level sampling for join size estimation.
\newblock In {\em SIGMOD}, pages 759--774, 2017.

\bibitem{CondieCAHES10}
T.~Condie, N.~Conway, P.~Alvaro, J.~M. Hellerstein, K.~Elmeleegy, and R.~Sears.
\newblock Mapreduce online.
\newblock In {\em NSDI}, pages 313--328, 2010.

\bibitem{considine2004approximate}
J.~Considine, F.~Li, G.~Kollios, and J.~Byers.
\newblock Approximate aggregation techniques for sensor databases.
\newblock In {\em ICDE}, pages 449--460, 2004.

\bibitem{ConsidineLKB04}
J.~Considine, F.~Li, G.~Kollios, and J.~W. Byers.
\newblock Approximate aggregation techniques for sensor databases.
\newblock In {\em ICDE}, pages 449--460, 2004.

\bibitem{dudley1999uniform}
R.~M. Dudley.
\newblock {\em Uniform central limit theorems}, volume~23.
\newblock Cambridge Univ Press, 1999.

\bibitem{FaloutsosRM94}
C.~Faloutsos, M.~Ranganathan, and Y.~Manolopoulos.
\newblock Fast subsequence matching in time-series databases.
\newblock In {\em SIGMOD}, pages 419--429, 1994.

\bibitem{GibbonsMP97}
P.~B. Gibbons, Y.~Matias, and V.~Poosala.
\newblock Fast incremental maintenance of approximate histograms.
\newblock In {\em VLDB}, pages 466--475, 1997.

\bibitem{GoldinK95}
D.~Q. Goldin and P.~C. Kanellakis.
\newblock On similarity queries for time-series data: Constraint specification
  and implementation.
\newblock In {\em CP}, pages 137--153, 1995.

\bibitem{HaasS92}
P.~J. Haas and A.~N. Swami.
\newblock Sequential sampling procedures for query size estimation.
\newblock In {\em SIGMOD}, pages 341--350, 1992.

\bibitem{hellerstein1997online}
J.~M. Hellerstein, P.~J. Haas, and H.~J. Wang.
\newblock Online aggregation.
\newblock In {\em SIGMOD Record}, volume~26, pages 171--182, 1997.

\bibitem{IoannidisP95}
Y.~E. Ioannidis and V.~Poosala.
\newblock Balancing histogram optimality and practicality for query result size
  estimation.
\newblock In {\em SIGMOD}, pages 233--244, 1995.

\bibitem{JermaineAPD07}
C.~M. Jermaine, S.~Arumugam, A.~Pol, and A.~Dobra.
\newblock Scalable approximate query processing with the {DBO} engine.
\newblock In {\em SIGMOD}, pages 725--736, 2007.

\bibitem{katsis2013delphi}
Y.~Katsis, C.~Baru, T.~Chan, S.~Dasgupta, C.~Farcas, W.~Griswold, J.~Huang,
  L.~Ohno-Machado, Y.~Papakonstantinou, F.~Raab, et~al.
\newblock Delphi: Data e-platform for personalized population health.
\newblock In {\em e-Health Networking, Applications \& Services (Healthcom),
  2013 IEEE 15th International Conference on}, pages 115--119. IEEE, 2013.

\bibitem{keogh1997fast}
E.~Keogh.
\newblock Fast similarity search in the presence of longitudinal scaling in
  time series databases.
\newblock In {\em ICTAI}, pages 578--584, 1997.

\bibitem{keogh2001locally}
E.~Keogh, K.~Chakrabarti, M.~Pazzani, and S.~Mehrotra.
\newblock Locally adaptive dimensionality reduction for indexing large time
  series databases.
\newblock {\em SIGMOD Record}, 30(2):151--162, 2001.

\bibitem{KeoghCPM01}
E.~J. Keogh, K.~Chakrabarti, M.~J. Pazzani, and S.~Mehrotra.
\newblock Dimensionality reduction for fast similarity search in large time
  series databases.
\newblock {\em KAIS}, 3(3):263--286, 2001.

\bibitem{KeoghP98}
E.~J. Keogh and M.~J. Pazzani.
\newblock An enhanced representation of time series which allows fast and
  accurate classification, clustering and relevance feedback.
\newblock In {\em KDD}, pages 239--243, 1998.

\bibitem{KeoghP99}
E.~J. Keogh and M.~J. Pazzani.
\newblock Relevance feedback retrieval of time series data.
\newblock In {\em SIGIR}, pages 183--190, 1999.

\bibitem{LazaridisM01}
I.~Lazaridis and S.~Mehrotra.
\newblock Progressive approximate aggregate queries with a multi-resolution
  tree structure.
\newblock In {\em SIGMOD}, pages 401--412, 2001.

\bibitem{madden2002tag}
S.~Madden, M.~J. Franklin, J.~M. Hellerstein, and W.~Hong.
\newblock Tag: A tiny aggregation service for ad-hoc sensor networks.
\newblock {\em SIGOPS}, 36(SI):131--146, 2002.

\bibitem{PansareBJC11}
N.~Pansare, V.~R. Borkar, C.~Jermaine, and T.~Condie.
\newblock Online aggregation for large mapreduce jobs.
\newblock {\em {PVLDB}}, 4(11):1135--1145, 2011.

\bibitem{PapapetrouGD12}
O.~Papapetrou, M.~N. Garofalakis, and A.~Deligiannakis.
\newblock Sketch-based querying of distributed sliding-window data streams.
\newblock {\em {PVLDB}}, 5(10):992--1003, 2012.

\bibitem{ShapiroC84}
G.~Piatetsky{-}Shapiro and C.~Connell.
\newblock Accurate estimation of the number of tuples satisfying a condition.
\newblock In {\em SIGMOD}, pages 256--276, 1984.

\bibitem{PoosalaGI99}
V.~Poosala, V.~Ganti, and Y.~E. Ioannidis.
\newblock Approximate query answering using histograms.
\newblock {\em {IEEE} Data Eng. Bull.}, 22(4):5--14, 1999.

\bibitem{PoosalaIHS96}
V.~Poosala, Y.~E. Ioannidis, P.~J. Haas, and E.~J. Shekita.
\newblock Improved histograms for selectivity estimation of range predicates.
\newblock In {\em SIGMOD}, pages 294--305, 1996.

\bibitem{PottiP15}
N.~Potti and J.~M. Patel.
\newblock {DAQ:} {A} new paradigm for approximate query processing.
\newblock {\em {PVLDB}}, 8(9):898--909, 2015.

\bibitem{ReissGH06}
F.~Reiss, M.~N. Garofalakis, and J.~M. Hellerstein.
\newblock Compact histograms for hierarchical identifiers.
\newblock In {\em VLDB}, pages 870--881, 2006.

\bibitem{ShatkayZ96}
H.~Shatkay and S.~B. Zdonik.
\newblock Approximate queries and representations for large data sequences.
\newblock In {\em ICDE}, pages 536--545, 1996.

\bibitem{terzi2006efficient}
E.~Terzi and P.~Tsaparas.
\newblock Efficient algorithms for sequence segmentation.
\newblock In {\em SDM}, pages 316--327, 2006.

\bibitem{Ting16}
D.~Ting.
\newblock Towards optimal cardinality estimation of unions and intersections
  with sketches.
\newblock In {\em SIGKDD}, pages 1195--1204, 2016.

\bibitem{WangS08}
H.~Wang and K.~C. Sevcik.
\newblock Histograms based on the minimum description length principle.
\newblock {\em {VLDB} J.}, 17(3):419--442, 2008.

\bibitem{WuOT10}
S.~Wu, B.~C. Ooi, and K.~Tan.
\newblock Continuous sampling for online aggregation over multiple queries.
\newblock In {\em SIGMOD}, pages 651--662, 2010.

\end{thebibliography}
}

\normalsize
\appendix
\section{Proofs}
\label{appendix:proofs}

\subsection{Error measures for the Times operator (Single Segment)}
\label{appendix:times}
 Let $f^{(1)}$ and $f^{(2)}$ be the compression functions of $T_1=(d_1^{(1)},...,d_n^{(1)})$ and $T_2=(d_1^{(2)},...,d_n^{(2)})$ respectively.  Let $(\mathcal{L}_1, d^*_1, f^*_1)$ and $(\mathcal{L}_2, d^*_2, f^*_2)$ be the error measures for time series $T_1$ and $T_2$. For \textsf{Times}($T_1,T_2$)$\rightarrow$ $T$ operator, the compression function $f$ and the error measures $(\mathcal{L}, d^*, f^*)$ for the output time series $T=(d_1,...,d_n)$ are computed as follows:
\begin{itemize}
  \item $f=f^{(1)}\times f^{(2)}$, i.e., the product of two compression functions.
  \item $\mathcal{L}=\sum_{i=1}^n|d_i-f(i)|=\sum_{i=1}^n|d_i^{(1)}d_i^{(2)}-f^{(1)}(i)f^{(2)}(i)|$. There are two options to transform this expression. \\
      Option 1: $\mathcal{L}=\sum_{i=1}^n|d_i^{(1)}d_i^{(2)}-d_i^{(1)}f^{(2)}(i)+d_i^{(1)}f^{(2)}(i)-f^{(1)}(i)f^{(2)}(i)|$ $=$ $\sum_{i=1}^n|d_i^{(1)} (d_i^{(2)}-f^{(2)}(i))+f^{(2)}(i)(d_i^{(1)}-f^{(1)}(i))|$ $\leq$ $f^*_2\mathcal{L}_1+d^*_1\mathcal{L}_2$.\\
      Option 2: $\mathcal{L}=\sum_{i=1}^n|d_i^{(1)}d_i^{(2)}-d_i^{(2)}f^{(1)}(i)+d_i^{(2)}f^{(1)}(i)-f^{(1)}(i)f^{(2)}(i)|$ $=$ $\sum_{i=1}^n|d_i^{(2)} (d_i^{(1)}-f^{(1)}(i))+f^{(1)}(i)(d_i^{(2)}-f^{(2)}(i))|$ $\leq$ $d^*_2\mathcal{L}_1+f^*_1\mathcal{L}_2$.\\
      Thus, we choose the minimal one between these two options. That is $\mathcal{L} = min\{f^*_2\mathcal{L}_1+d^*_1\mathcal{L}_2, d^*_2\mathcal{L}_1+f^*_1\mathcal{L}_2\}$.
  \item $d^*=max\{|d_i|~|1\leq i\leq n\}=max\{|d_i^{(1)}d_i^{(2)}|~|1\leq i\leq n\}$$\leq$ $d^*_1\times d^*_2$.
  \item $f^*=max\{|f(i)|~|1\leq i\leq n\}=max\{|f_i^{(1)}f_i^{(2)}|~|1\leq i\leq n\}$$\leq$ $f^*_1\times f^*_2$.
\end{itemize}

\subsection{Proof of the optimality of the error estimation formulas of Figure \ref{fig:error_estimators}}
\label{appendix:lower_bound}
\noindent\textbf{\underline{Aggregation operator.}}
Depending on whether it is a single segment time series, there are two cases.

\noindent\textbf{Case 1}. A time series $T$ (with $n$ data points) contains only one single segment. There are two subcases depending on whether $T$ is entirely used in the query or not. That is $\ell_e-\ell_s=n$ or not.

\textbf{Case 1.1}. $\ell_e-\ell_s=n$. In this case, the error $\varepsilon=\sum_{i=\ell_s}^{\ell_e}|d_i-f(i)=\sum_{i=1}^{n}|d_i-f(i)|$. And we have $\mathcal{L}=\sum_{i=1}^{n}|d_i-f(i)|$.
Therefore, we can get $\varepsilon=\mathcal{L}$. It means that by using $\mathcal{L}$ we are able to get the accurate error (the optimal error estimation). As desired.

\textbf{Case 1.2}. $\ell_e-\ell_s<n$.  Assume there exists an error estimator $A$ that gives an approximate error $\tilde{\varepsilon}$, where $\tilde{\varepsilon}=\mathcal{L}-\alpha$, where $\alpha>0$ is a small value. Let $T$ be time series segment with length $n$ such that $d_i=f(i)$ for $i\in[1, n-(\ell_e-\ell_s)-1]$ and $d_i=f(i)+1$ for $i\in[n-(\ell_e-\ell_s),n]$. Thus $\mathcal{L}=\sum_{i=1}^{n}|d_i-f(i)| = \ell_e-\ell_s$. For a query with range $[1, \ell_e-\ell_s]$, $A$ gives the approximate error as $\ell_e-\ell_s-\alpha$, which is correct as the accurate error is $0$. Now we switch the points $d_i$ with $d_j$ (as well as $f(i)$ with $f(j)$) for $i\in [1,\ell_e-\ell_s], j\in [n-(\ell_e-\ell_s),n]$ to generate a new segment $T'$. Note that, $T'$ and $T$ have the same $\mathcal{L}$. Now the accurate error is also $\mathcal{L}=\ell_e-\ell_s$. However, $A$ still gives the approximate error as $\ell_e-\ell_s-\alpha$, which is incorrect as it produces smaller error than the optimal one. Therefore, there does not exist an estimator that produces approximate errors less than that of our estimator, which means our estimator achieves the lower bound. As desired.

\noindent\textbf{Case 2}.  A time series $T$ contains multiple segments $S_1,..., S_n$. Note that, segments $S_2,...,S_{n-1}$ are all always entirely used by the query. According to case 1.1, our estimator gives the optimal error estimation. For the left-most and right-most segments, i.e., $S_1$ and $S_n$, our estimator achieves the lower bound of error estimation according to case 1.2. As desired.

\smallskip
\noindent\textbf{\underline{Plus and Minus operators.}}
Similar to the proof presented above, it is easy to see our estimator achieves the lower bound. Otherwise, there must exist an incorrect error estimation.

\smallskip
\noindent\textbf{\underline{Times operator.}}
We distinguish between two cases.\\
\noindent\textbf{Case 1}. Time series $T_1$ with $n$ data points (resp. $T_2$) only contains one segment. Depending on whether  $T_1$ and $T_2$ are entirely used. There are two subcases.

\textbf{Case 1.1}. $T_1$ and $T_2$ are entirely used. The accurate error is $\varepsilon=\sum_{i=1}^{n}(d_i^{(1)}\times d_i^{(2)} - f^{(1)}(1)\times f^{(2)}(i))$. Let the data point in $T_1$ have the following features $d_i^{(1)}=d_{i+1}^{(1)}-1$, $f^{(1)}(i)=f^{(1)}(i+1)-1$   and $f^{1}(i)=d_i^{(1)}-1$ for $i\in[1,n]$. Let  $T_2$ have the same data. Thus, $\mathcal{L}^{(1)}=\sum_{i=1}^{n}|d_i^{(1)}-f^{(1)}(i)| = n$ and $\mathcal{L}^{(2)} = n$, $d_1^*=d_n^{(1)}=d_2^{*}=d_n^{(2)}$ and $f_1^*=f^{(1)}(n)=f_2^*=f^{(2)}(n)$. So the estimated error of our estimator is \[\hat{\varepsilon}=min\{f^*_2\mathcal{L}_1+d^*_1\mathcal{L}_2, d^*_2\mathcal{L}_1+f^*_1\mathcal{L}_2\}\]
\[=n(d_n^{(1)}+\hat{d}_n^{(1)})=n(2\hat{d}_n^{(1)}+1)\]

Assume there exists an error estimator $A$ that produces an approximate error $\hat{\varepsilon}'$, where $\hat{\varepsilon}'=\hat{\varepsilon}-\alpha$, where $\alpha>0$ is a small value. Since the accurate error $\varepsilon=\sum_{i=1}^{n}(2f{(1)}(i)+1)<\hat{\varepsilon}'$. $A$ returns the correct estimation. Then we make $d_i^{(1)}=d_{i+1}^{(1)}$ and $f{(1)}(i)=f^{(1)}(i+1)$  for $i\in[1,n]$ in both two time series. Then the error measures of both time series stay the same. Now, the accurate error is  $\varepsilon=\sum_{i=1}^{n}(2\hat{d}_i^{(1)}+1)=n(2\hat{d}_n^{(1)}+1)=\hat{\varepsilon}$. That is our estimator produces the optimal error (meaning the error is equal to the accurate one). But $A$ still gives the same estimation $\hat{\varepsilon}-\alpha$, which is incorrect. So $A$ does not exist. As desired.

\textbf{Case 1.2}.  $T_1$ and $T_2$ are not entirely used. The proof is similar to the case 1.2 in the proof of aggregation operator above by making segments that $d_i^{(1)}\neq f^{(1)}(i)$ and $d_i^{(2)}\neq f^{(2)}(i)$ only happen in the query range, which makes $\hat{\varepsilon}=\varepsilon$. So our estimator achieves the lower bound.

\noindent\textbf{Case 2}. $k$ segments $S_1^{(1)},...,S_k^{(1)}$ in time series $T_1$ overlapping one segment $S^{(2)}$ in time series $T_2$ are used in the query with range [a,b]. The proof is similar to that in Case 1.1 and Case 1.2.

 \noindent\textbf{Case 3}. $k_1$ segments $S_1^{(1)},...,S_{k_1}^{(1)}$ in time series $T_1$ overlapping $k_2$ segments  $S_1^{(2)},...,S_{k_2}^{(2)}$  in time series $T_2$ are used in the query. The proof is similar to that in Case 2.

\smallskip
\noindent\textbf{\underline{Arithmetic Operators}}
For arithmetic operator $Ar_1\otimes\ Ar_2$. There are three cases depending on the number of approximate answers $Agg$ in the expression.

\noindent\textbf{Case 1}. Zero $Agg$, i.e., both $Ar_1$ and $Ar_2$ are numbers.  $Ar_1\otimes Ar_2$ can be transformed as $number_1\otimes number_2$ Then the answers are accurate answers. As desired.

\noindent\textbf{Case 2}. One $Agg$, i.e., $Ar_1\otimes Ar_2$ can be transformed as $Agg\otimes number$.  Let $\hat{Agg}$ and $\hat{\varepsilon}$ be the output approximate answer of $Agg$ and the estimated error by \db. Therefore, we know that $\hat{\varepsilon}$ is the lower bound of $|Agg-\hat{Agg}|$. For $Agg+number$, the error is $|Agg+number-(\hat{Agg}+number)|=|Agg-\hat{Agg}|$. Thus,   $\hat{\varepsilon}$ is also the lower bound of $Agg+number$. Similarly, we can prove the lower bound property of the errors for $\{-,\times,\div\}$ operators.

\noindent\textbf{Case 3}. Two $Agg$, i.e., $Ar_1\otimes Ar_2$ can be transformed as $Agg_1\otimes Agg_2$. Let $\hat{Agg}_1$ (resp. $\hat{Agg}_2$)  and $\hat{\varepsilon}_1$ (resp. $\hat{\varepsilon}_2$) be the output approximate answer of $Agg_1$ (resp. $Agg_2$) and the estimated errors provided by \db, respectively. According to the previous proof, we know that $\hat{\varepsilon}_1$ and $\hat{\varepsilon}_2$ are the lower bound of $|Agg_1-\hat{Agg}_1|$ and $|Agg_2-\hat{Agg}_2|$ respectively. It is obvious that $\hat{\varepsilon}_1+\hat{\varepsilon}_2$ is the lower bound error of  $Agg_1+Agg_2$ and $Agg_1-Agg_2$. For $Agg_1\times Agg_2$, we have $|Agg_1\times Agg_2- \hat{Agg}_1\times \hat{Agg}_2|\leq \hat{Agg}_1\hat{\varepsilon}_2+\hat{Agg}_2\hat{\varepsilon}_1+\hat{\varepsilon}_1\hat{\varepsilon}_2$. We can prove it is the lower bound by constructing a case that the accurate error is equals to this one, which means there does not exist a better estimation. Similarly, for $Agg_1\div Agg_2$, we can prove $\frac{\hat{Agg}_1+\hat{\varepsilon}_1}{\hat{Agg}_2-\hat{\varepsilon}_2}-\frac{\hat{Agg}_1}{\hat{Agg}_2}$ is lower bound error.

\end{document}